\documentclass[reprint,amsfonts, amssymb, amsmath,  showkeys,pra, superscriptaddress, twocolumn,longbibliography]{revtex4-2}
\usepackage{float}
\makeatletter
\let\newfloat\newfloat@ltx
\makeatother
\usepackage[english]{babel}
\usepackage[utf8]{inputenc}
\usepackage{graphics}
\usepackage{selinput}
\usepackage[normalem]{ulem}
\usepackage[shortlabels]{enumitem}

\usepackage{braket}
\usepackage{amsthm}
\usepackage{mathtools}
\usepackage{physics}
\usepackage{xcolor}
\usepackage{graphicx}
\usepackage[left=16mm,right=16mm,top=35mm,columnsep=15pt]{geometry} 
\usepackage{adjustbox}
\usepackage{placeins}
\usepackage[T1]{fontenc}
\usepackage{lipsum}
\usepackage{csquotes}
\usepackage{bm}

\usepackage[linesnumbered,ruled,vlined]{algorithm2e}
\SetKwInput{kwInit}{Init}

\def\HC{\mathcal{H}}

\def\ad{^{\dagger}}

\newcommand{\fsnull}[1]{}
\newcommand{\old}[1]{}

\usepackage[makeroom]{cancel}
\usepackage[toc,page]{appendix}
\usepackage[colorlinks=true,citecolor=blue,linkcolor=magenta]{hyperref}

\usepackage{tikz}
\tikzset{every picture/.style=remember picture}

\newcommand{\dya}[1]{\ket{#1}\!\bra{#1}}

\newcommand{\Ubb}{\mathbb{U}}
\newcommand{\Rbb}{\mathbb{R}}

\newcommand{\AC}{\mathcal{A}}
\newcommand{\BC}{\mathcal{B}}
\newcommand{\CC}{\mathcal{C}}

\newcommand{\EC}{\mathcal{E}}

\newcommand{\GC}{\mathcal{G}}

\newcommand{\OC}{\mathcal{O}}

\newcommand{\RC}{\mathcal{R}}
\newcommand{\SC}{\mathcal{S}}

\newcommand{\YC}{\mathcal{Y}}

\newcommand{\Var}{{\rm Var}}

\renewcommand{\geq}{\geqslant}
\renewcommand{\leq}{\leqslant}

\renewcommand{\vec}[1]{\boldsymbol{#1}}  

\newcommand*{\id}{\openone}

\newcommand{\bs}{\textsf{BS}}

\newcommand{\tin}{{\text{in}}}

\newcommand{\thv}{\vec{\theta}}

\newcommand{\losalamos}{Theoretical Division, Los Alamos National Laboratory, Los Alamos, New Mexico 87545, USA}

\def\be{\begin{equation}}
\def\ee{\end{equation}}
\def\bs{\begin{split}}
\def\e{\end{split}}
\def\ba{\begin{eqnarray}}
\def\bea{\begin{eqnarray}}

\def\tea{\end{eqnarray}}
\def\ea{\end{eqnarray}}
\def\eea{\end{eqnarray}}

\newcommand\mf[1]{\mathfrak{#1}}

\newtheorem{theorem}{Theorem}
\newtheorem{lemma}{Lemma}
\newtheorem{proposition}{Proposition}

\newtheorem{definition}{Definition}
\newtheorem{hypothesis}{Hypothesis Class}

\usepackage{amssymb}
\usepackage{dsfont}

\def\be{\begin{equation}}
\def\te{\end{equation}}
\def\ee{\end{equation}}
\def\ba{\begin{eqnarray}}
\def\bea{\begin{eqnarray}}

\def\tea{\end{eqnarray}}
\def\ea{\end{eqnarray}}
\def\eea{\end{eqnarray}}

\begin{document}

\title{ Group-Invariant Quantum Machine Learning}

\author{Mart\'{i}n Larocca}
\thanks{The two first authors contributed equally.}
\affiliation{\losalamos}
\affiliation{Center for Nonlinear Studies, Los Alamos National Laboratory, Los Alamos, New Mexico 87545, USA}

\author{Fr\'{e}d\'{e}ric Sauvage}
\thanks{The two first authors contributed equally.}
\affiliation{\losalamos}

\author{Faris M. Sbahi}
\affiliation{X, Mountain View, CA, 94043, USA} \affiliation{Google Brain, Mountain View, CA, 94043, USA}

\author{Guillaume Verdon}
\affiliation{X, Mountain View, CA, 94043, USA}
\affiliation{Institute for Quantum Computing, University of Waterloo, ON, Canada}
\affiliation{Department of Applied Mathematics, University of Waterloo, ON, Canada}

\author{Patrick J. Coles}
\affiliation{\losalamos}
\affiliation{Quantum Science Center, Oak Ridge, TN 37931, USA}

\author{M. Cerezo}
\email{cerezo@lanl.gov} 
\affiliation{Information Sciences, Los Alamos National Laboratory, Los Alamos, NM 87545, USA}
\affiliation{Quantum Science Center, Oak Ridge, TN 37931, USA}

\begin{abstract}
Quantum Machine Learning (QML) models are aimed at learning from data encoded in quantum states. Recently, it has been shown that models with little to no inductive biases (i.e., with no assumptions  about the problem embedded in the model) are likely to have trainability and generalization issues, especially for large problem sizes. As such, it is fundamental to develop schemes that encode as much information as available about the problem at hand.  In this work we present a simple, yet powerful, framework where the underlying invariances in the data are used to build QML models that, by construction, respect those symmetries. These so-called group-invariant models produce outputs that remain invariant under the action of any element of the symmetry group $\mathfrak{G}$ associated to the dataset. 
We present theoretical results underpinning the design of $\mathfrak{G}$-invariant models, and exemplify their application through several paradigmatic QML classification tasks including cases when $\mathfrak{G}$ is a continuous Lie group and also when it is a discrete symmetry group. Notably, our framework allows us to recover, in an elegant way, several well known algorithms for the literature, as well as to discover new  ones.  Taken together, we expect that our results will help pave the way towards a more geometric and group-theoretic  approach to QML model design.
\end{abstract}

\maketitle


\section{Introduction}

Symmetries have always held a special place in the imaginarium of scientists seeking to understand the universe through physical theories. As such, it is not strange for a scientist to equate a theory's  beauty and elegance with its symmetry and harmony~\cite{gellmann2007beauty}. Still, the role of symmetries in science is more than simply aesthetic, as in many cases they constitute the underlying force behind a theory. For instance, Galilean invariance is  pivotal in Newton's laws of motion~\cite{newton1999principia}, and Lorentz and gauge invariances were fundamental for Maxwell to unify electricity and magnetism  into the general theory of electromagnetism~\cite{maxwell1865viii}. In the 20th century, symmetries would take the center stage as Einstein's theory of general relativity provided the first geometrization of symmetries~\cite{einstein1922general,gross1996role}. Soon after, Noether's theorem showed a  connection between differentiable symmetries and conserved quantities~\cite{noether1918invariante}, proving that  symmetries have defining implications in  nature.

More recently, the importance of symmetries has been explored in the context of machine learning, and is core to the development of the field of \textit{geometric deep learning}~\cite{bronstein2021geometric}. Here, the key insight was to note that the most successful neural network architectures can be viewed as models with inductive biases that respect the underlying structure and symmetries of the domain over which they act.  The inductive bias refers to the fact that the model explores only a subset of the space of functions due to the assumptions imposed on its definition. Geometric deep learning not only constitutes a unifying mathematical framework for studying neural network architectures, but also provides guidelines to incorporate prior physical (and geometrical) knowledge into new architectures with better generalization performance, more efficient data requirements, as well as favorable optimization landscapes~\cite{cohen2016group,kondor2018generalization,bronstein2021geometric,bogatskiy2022symmetry}.

In this work, we propose to import ideas from the field of geometric deep learning to the realm of Quantum Machine Learning (QML). QML has recently emerged as a leading candidate to make practical use of near-term quantum devices~\cite{biamonte2017quantum,cerezo2020variationalreview}. QML formally generalizes classical machine learning by embedding it into the formalism of quantum mechanics and quantum computation. This formal generalization can lead to practical speedups with the potential to significantly outperform classical machine learning~\cite{huang2021provably}, since quantum computers can efficiently manipulate information stored in quantum states living in exponentially large Hilbert spaces.

Similar to their classical counterparts, the ability of QML models to solve a given task hinges on several factors, with one of the most important being the choice of the model itself. If the inductive biases~\footnote{This is also known as the choice and parameterization of \textit{prior}, in the language of Bayesian theory \cite{battaglia2018relational}.} of a model are uninformed, its expressibility is large, leading to issues such as barren plateaus in the training landscape~\cite{holmes2021connecting,mcclean2018barren,cerezo2020cost,sharma2020trainability,thanasilp2021subtleties,arrasmith2021equivalence}. Adding sharp priors to the model narrows the effective search space, enhancing its trainability and improving its generalization performance~\cite{cong2019quantum,pesah2020absence,volkoff2021large,caro2021generalization}. As such, a great deal of effort has been recently put forward towards designing more problem-specific schemes with strong inductive biases~\cite{larocca2021diagnosing,larocca2021theory,wecker2015progress,gard2020efficient,lee2021towards,tang2019qubit,glick2021covariant,verdon2019quantumgraph,verdon2019quantum}. Despite these efforts, most architectures currently used in the literature remain problem-agnostic, as there is no overarching theoretical framework that provides guidelines for how to embed symmetries of the problem into the QML model.

The main contribution of this work is a series of Propositions (Propositions~\ref{prop:no-post-processing-symmetries},~\ref{prop:no-post-processing-orthogonal},~\ref{prop:no-post-processing-symmetries-two} and~\ref{prop:no-post-processing-orthogonal-two}) characterizing the landscape of possible quantum machine learning models that are invariant under a given symmetry group $\mf{G}$. By identifying different types of strategies leading to group invariant models, we pave the way for a more systematic and efficient symmetry-informed QML model design. The power of our framework is showcased by applying it to classifying datasets based on purity,  time-reversal dynamics, multipartite entanglement, and graph isomorphism, where we are able to recover, in an effortless manner, many celebrated quantum protocols and algorithms including very recent ones~\cite{buhrman2001quantum,cotler2021revisiting,beckey2021computable,brennen2003observable,meyer2002global,rungta2001universal,bhaskara2017generalized,carvalho2004decoherence,wong2001potential,hen2012solving,gaitan2014graph,zick2015experimental,izquierdo2020discriminating}. 
We also discuss the extension of our framework to the design of equivariant quantum neural networks. 
Finally, we highlight the exciting outlook for the new field of \textit{geometric quantum machine learning}, for which our article lays some of the groundwork.

\begin{figure}[t!]
	\includegraphics[width= .9 \columnwidth]{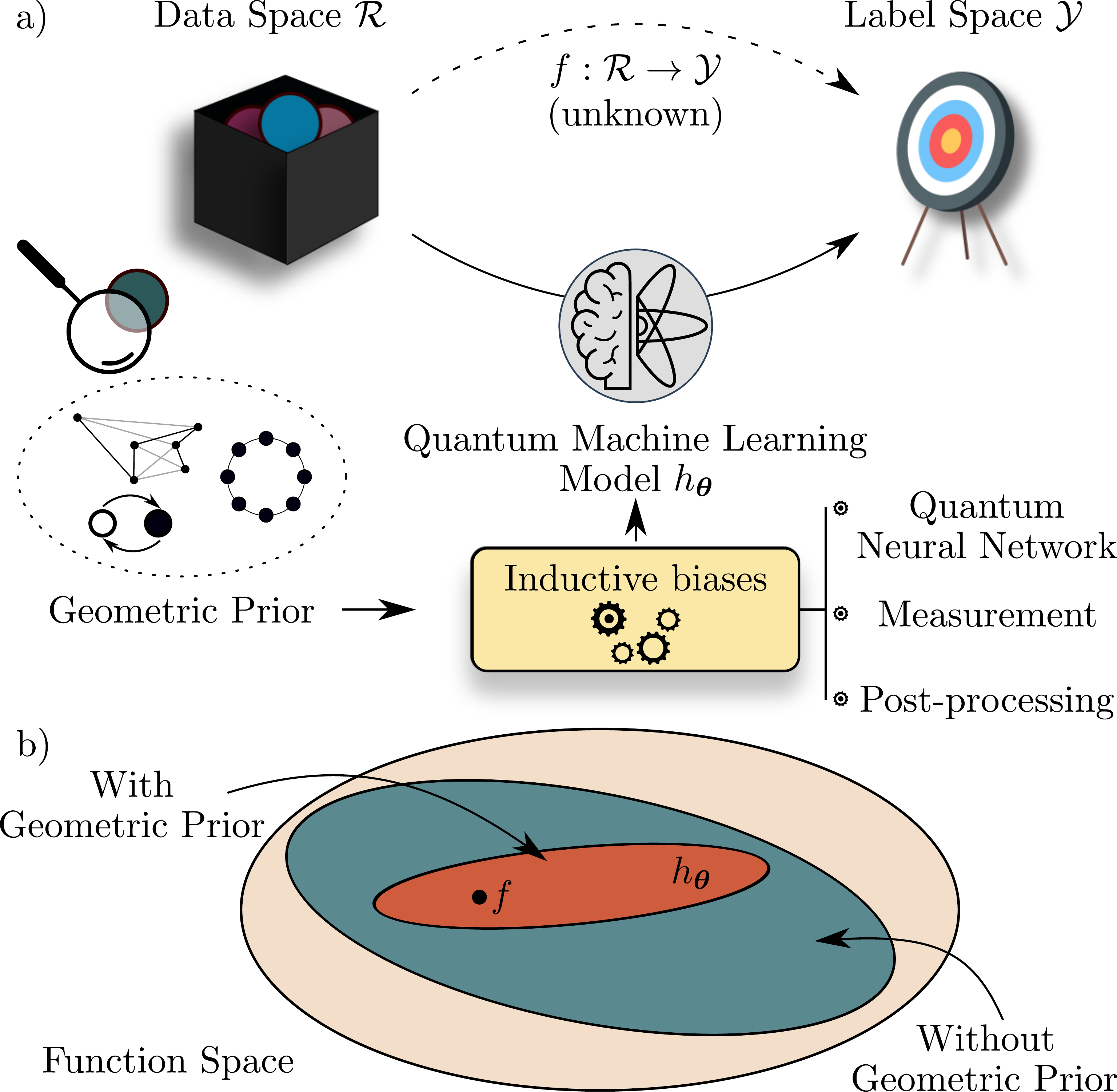}
	\caption{\textbf{The role of inductive biases.} a) In a quantum supervised learning task, the goal is to train a parameterized model $h_{\thv}$ to match the label predictions of an unknown function $f$. The inductive biases represent the assumptions (priors) about our knowledge of how the inputs  are related to the output predictions. These are encoded in the way the model is built. The goal of this work is to determine geometric priors from the underlying symmetries in the quantum data. b) Due to the biases in the model,  $h_{\thv}$ explores only a subset of all possible functions. }
	\label{fig:1}
\end{figure}

\section{Preliminaries}\label{section:framework}

Here we provide the background and definitions needed for our group-invariant QML framework.

\subsection{Symmetry groups in supervised QML}

In this work we consider supervised binary classification tasks on quantum data. We remark, however, that the methods derived here can be readily applied to   more general supervised learning scenarios or to unsupervised learning tasks. In addition, our work applies to classical data that has been encoded into quantum states.

For our purposes, we consider the case where one is given repeated access to a  set of labeled training data from a dataset of the form $\SC=\{(\rho_i,y_i)\}_{i=1}^N$. Here, $\rho_i$ are $n$-qubit quantum states from a data domain $\RC$  in a $d$-dimensional Hilbert space (with $d=2^n$), while  $y_i$ are binary labels from a label domain $\YC=\{0,1\}$. The data in $\SC$ is drawn i.i.d. from a distribution defined over $\RC\times\YC$, and we assume that the labels associated to each quantum state are assigned according to some (unknown) function $f:\RC\rightarrow \YC$, so that $f(\rho_i)=y_i$. As shown in Fig.~\ref{fig:1}(a), the goal is  to train a  parameterized model (or hypothesis) $h_{\thv}$ to produce labels that match those of the target function $f$ with high probability. Here $\thv$ denotes the set of trainable parameters in the model.

For a given dataset, a fundamental question to ask is: \textit{what is the set of unitary operations on the states $\rho_i$ that leave their respective labels $y_i$ unchanged?} Such set of operations forms a group~\footnote{We make here two important remarks. First, $\mathfrak{G}$ must form a group as composing two symmetries leads to a new symmetry. Similarly, symmetric transformations are always invertible and their inverse is a symmetry itself. Second, in more precise terms, $\mathfrak{G}$ is a unitary representation of a group.} $\mathfrak{G} \subseteq \mathbb{U}(d)$, a subset of $\mathbb{U}(d)$ the unitary group of degree $d$. In the following, $\mathfrak{G}$ is referred to as the symmetry group of the dataset. Explicitly, for every element $V$ in $\mathfrak{G}$, the label associated with any transformed state $V\rho_i V\ad$ is exactly the same as the label associated with the original $\rho_i$. Hence, it is natural to require that the model we are training should produce labels that also remain invariant under the action of $\mathfrak{G}$ on the data. To capture such invariance, we introduce the following definition.
\begin{definition}[$\mathfrak{G}$-invariance]\label{def:G-invariance}
    A function $h$ is $\mathfrak{G}$-invariant iff 
    \begin{equation}
        h(V \rho V\ad)=h(\rho)\quad \text{for all $V\in \mathfrak{G}$, $\rho\in\RC$.}
    \end{equation} 
\end{definition}

In principle, such $\mathfrak{G}$-invariance could be heuristically learnt by $h_{\thv}$ via data-augmentation~\cite{bekkers2018roto}, i.e., by including additional training instances of the form $\{V\rho_i V\ad,y_i\}$. However, such effort is undesirable for two main reasons. First, it obviously means an increased algorithmic run-time cost. But second, and most importantly, such invariance learning is not guaranteed to be completely successful (especially when $\mathfrak{G}$ is large or  continuous)~\cite{bronstein2021geometric}. Instead, as shown in Fig.~\ref{fig:2}, our main approach here is to design QML models $h_{\thv}$ that are, \textit{by construction}, $\mathfrak{G}$-invariant for all $\thv$. To achieve this, we introduce biases in the structure of $h_{\thv}$, for instance, by carefully choosing the architecture of the quantum neural network employed and the physical observable measured. 

\begin{figure}[t!]
	\includegraphics[width= 1\columnwidth]{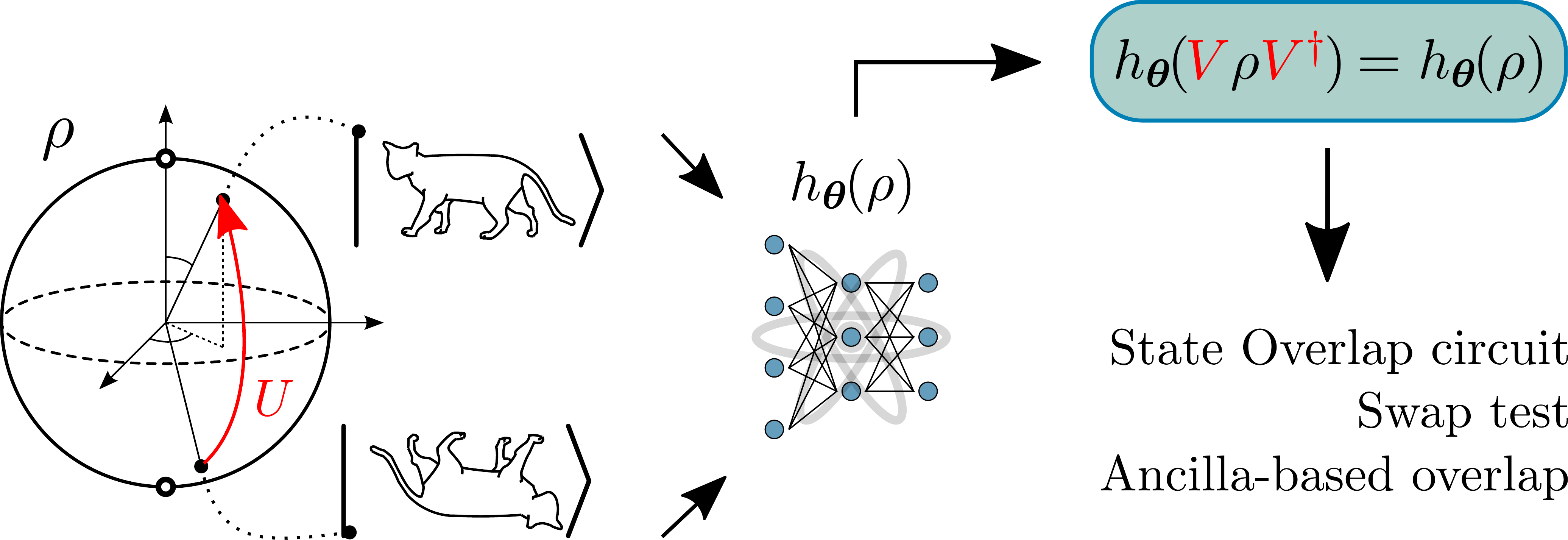}
	\caption{\textbf{Group invariance in QML.} Consider a QML task of classifying single-qubit states according to their purity. The symmetry group  $\mathfrak{G}$ of this task is the unitary group $\mathbb{U}(2)$, as any unitary $V\in\mathbb{U}(2)$  preserves the spectral properties of the data. By imposing  $\mathfrak{G}$-invariance into the quantum model, so that $h_{\thv}(V\rho V\ad)=h_{\thv}(\rho)$ for all $V\in\mathbb{U}(2)$, we can rediscover several known algorithms such as those listed in the figure.}
	\label{fig:2}
\end{figure}

In the context of binary classification, there are two main scenarios to consider. In a first scenario, the data in both classes is invariant under a same symmetry group $\mathfrak{G}$, and thus $h_{\thv}$ needs to be invariant under this sole symmetry group. 
In the second scenario the data in different classes have different symmetries, and we denote as $\mathfrak{G}_0$ and $\mathfrak{G}_1$ the symmetry groups corresponding to data with labels $y_i=0$ and $y_i=1$, respectively. 
In such case, we have the freedom to consider QML models that are either $\mathfrak{G}_0$-invariant, $\mathfrak{G}_1$-invariant, or both. As shown below, it is often convenient to build models that are invariant only under the action of one of the symmetry groups, as this is sufficient for data classification.

\subsection{Conventional and quantum-enhanced experiments}

\begin{figure*}[t!]
	\includegraphics[width= 1 \linewidth]{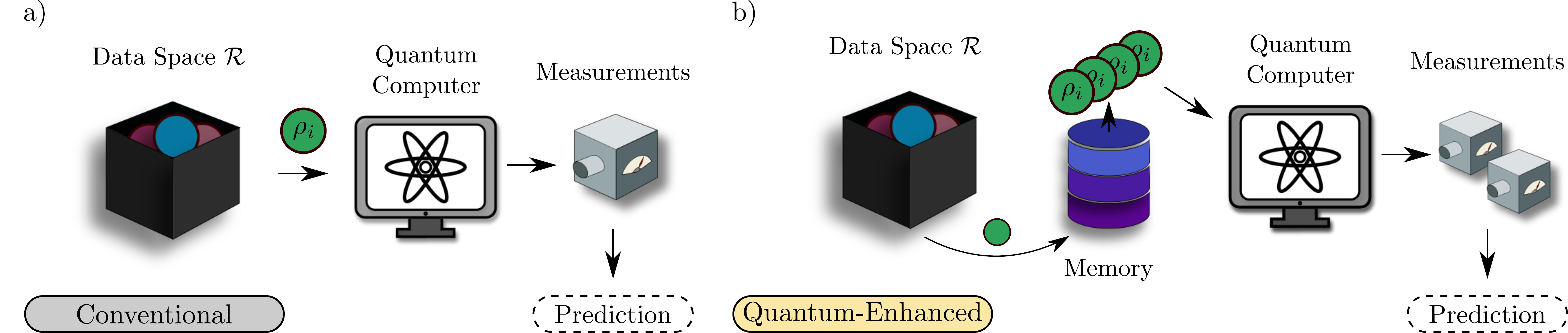}
	\caption{\textbf{Conventional and quantum-enhanced experiments.} a) In a conventional experiment, each data instance $\rho_i$ in the dataset is sent to a quantum device. The state obtained at the output of the quantum computation is measured, with the measurement outcomes used to make predictions. b) In a quantum-enhanced experiment,  a quantum memory allows us to store several copies of each state $\rho_i$ in the dataset. These copies can be simultaneously sent to a quantum device. The state obtained at the output of the quantum computation is measured, with the measurement outcomes used to make predictions.   }
	\label{fig:3}
\end{figure*}

Thus far, we have not defined what constitutes the parameterized model $h_{\thv}$. 
This choice is tied to the physical resources one may have access to, i.e., the way quantum data can be stored, accessed and measured. Since there is a large amount of freedom in this regard, we find it useful to restrict ourselves to two scenarios. Following the demarcation proposed in \cite{huang2021quantum,aharonov2022quantum} we consider two settings, a conventional and a quantum-enhanced one, which are defined as:
\begin{definition}[Conventional experiment]\label{def:conventional_experiment}
In conventional experiments, each data instance $\rho_i$ is processed in a quantum computer and measured individually.
\end{definition}
\begin{definition}[Quantum-enhanced experiment]\label{def:quantum_enhanced}
In quantum-enhanced experiments, multiple copies of each data instance $\rho_i$ can be stored in a quantum memory, and later simultaneously processed and measured in a quantum computer. 
\end{definition}

In both settings, the model predictions are obtained from the quantum device experiment outcomes.
However, as illustrated in Fig.~\ref{fig:3}, the key difference between the classical and quantum-enhanced settings is that, in the latter, the QML model is allowed to act coherently on multiple copies of $\rho_i$. This is in contrast with the conventional setting where the QML model can only operate over a single copy of $\rho_i$ at a time.

\subsection{QML model structure}

Throughout this work we consider models consisting in a quantum neural network $U(\thv)$ (i.e., a parameterized unitary that can be realized on a quantum computer) operating on $k$ copies of an input state $\rho$, followed by a measurement on the resulting state. 
In other words, we work with models belonging to the following hypothesis class.
\begin{hypothesis}~\label{def:hyp_class_1}
We define the Hypothesis Class $\HC_1$ as composed of functions of the form
\begin{align}\label{eq:linear_cost}
    h_{\thv}^{(k)}(\rho)=\Tr[U(\thv)(\rho^{\otimes k}) U\ad(\thv)O]\,,
\end{align}
where $k$ is the number of copies of the data  state $\rho$, $U(\thv)$ is a quantum neural network, and $O$ is a Hermitian operator.
\end{hypothesis}
For $k=1$ copies, the models belonging to the Hypothesis Class $\HC_1$ correspond to those that can be computed on a conventional experiment according to Definition~\ref{def:conventional_experiment}. On the other hand, $k\geq 2$ copies lead to models in quantum-enhanced experiment according to Definition~\ref{def:quantum_enhanced}. 

Arguably, models from the Hypothesis Class~\ref{def:hyp_class_1} are not of the most general form. 
For instance, these could be extended to allow for non-trivial classical post-processing of the measurement outcomes and also to involve more than one circuit or observable.
Still, the Hypothesis Class~\ref{def:hyp_class_1} already encompasses most of the current QML frameworks \cite{cerezo2020variational} and can serve as a basis for more expressive QML models. In the following we restrict our attention to models pertaining to $\HC_1$ and leave the study of more general models for future work. 

\subsection{Classification accuracy}

Let us define some terminology that will allow us to assess the accuracy of a model's classification. First, we remark that we do not consider precision issues when discussing classification accuracy. Recall that the model's predictions in Eq.~\eqref{eq:linear_cost} are expectation values, which in practice need to be estimated via measurements on a quantum computer. Hence, given a finite number of shots (measurement repetitions), these can only be resolved up to some additive errors. However, for the sake of simplicity,  we here assume the limit of zero shot noise (i.e., infinite precision), and we will challenge this assumption when appropriate in the results section. With this remark in hand, consider the following definitions of different degrees of classification accuracy.
\begin{definition}[Classification Accuracy]\label{def:accuracy}
i) We say that a model provides no information that can classify the data if its outputs are always the same irrespective of the label associated to the input quantum state.  ii) We say that a model performs noisy classification if its outputs are the same for some, but not all, data in different classes. iii) We say that a model  perfectly classifies the data if its outputs are never the same for data in different classes.
\end{definition}

We note that in some cases a model can, at best, only perform noisy classification as its accuracy will be fundamentally limited by the distinguishability of the quantum states in the dataset. Note that this is typically not an issue for classical datasets, although the issue does arise for noisy classical data. In some cases, for example as in the time-reversal dataset that we consider below, the quantum data states associated with different output labels are non-orthogonal. In this case, perfect classification cannot be achieved, regardless of the form of the model.

\subsection{Useful definitions}

In this, rather mathematical, section we present definitions that will be used throughout the main text. For further reading we refer  to~\cite{kirillov2008introduction,zeier2011symmetry}. 

While $\mathfrak{G}$ describes the symmetries in the data, it will be crucial to characterize the symmetries of $\mathfrak{G}$ itself. 
The symmetries of a group are captured by the \textit{commutant} 
\begin{equation}
    \CC(\mathfrak{G})=\{W\in\mathbb{C}^{d\times d}\quad|\quad[W,V]=0\,,\quad \forall V\in \mathfrak{G}\}\,,
\end{equation}
which is the vector space of all $d\times d$ complex matrices that commute with $\mathfrak{G}$. More generally one can also consider the space of matrices in $\mathbb{C}^{d^k \times d^k}$  commuting with  the $k$-th power tensor of the elements in $\mathfrak{G}$~\cite{zimboras2015symmetry}.

\begin{definition}[$k$-th order symmetries]\label{def:symmetries}
Given a unitary representation $\mathfrak{G}\subseteq\mathbb{U}(d)$ of a group, its $k$-th order symmetries are
\begin{equation}\label{eq:k_symmetries}
    \CC^{( k)}(\mathfrak{G})=\{W\in \mathbb{C}^{d^k \times d^k}\,\,|\,\,[W,V^{\otimes k}]=0\,,\quad \forall V\in \mathfrak{G}\},
\end{equation}
for all positive integers $k$.
\end{definition}
First-order symmetries ($k=1$) are known as the \textit{ linear symmetries}, while second order ones ($k=2$) are known as \textit{quadratic symmetries} of $\mathfrak{G}$. In general,  there may be $k$-th order symmetries that are not Hermitian (and thus, not physical observables). However, as proved in Appendix~\ref{app:A}, any matrix in $\CC^{( k)}(\mathfrak{G})$ has non-zero projection into the Hermitian subspace of $\CC^{(k)}(\mathfrak{G})$. Hence, one can always associate any non-Hermitian  element in  $\CC^{(k)}(\mathfrak{G})$ to a Hermitian one that also belongs in $\CC^{(k)}(\mathfrak{G})$.

While the $k$-th order symmetries can be defined for any group, in the case when $\mathfrak{G}$ is a Lie group there exists additional structure that one can exploit. In particular, there exists an associated Lie algebra $\mathfrak{g}\subseteq\mathfrak{u}(d)$  such that $e^{\mathfrak{g}}=\mathfrak{G}$. That is, $\mathfrak{g}=\{g\in\mathfrak{u}(d) \,|\, e^g \in\mathfrak{G} \}$. Here, $\mathfrak{u}(d)$ denotes  the set of $d\times d$ skew-symmetric matrices. We also find it convenient to introduce the following definition:
\begin{definition}[Orthogonal complement]\label{def:orthogonal-complement}
Given a Lie algebra $\mathfrak{g}\subseteq \mathfrak{u}(d)$, its orthogonal complement with respect to the Hilbert-Schmidt norm is defined as 
\begin{align}\label{eq:orthogonal_complement}
    \mathfrak{g}^{\perp}=\{h\in\mathfrak{u}(d)\quad|\quad\Tr[h\ad g]=0\,,\quad \forall g\in \mathfrak{g}\}.
\end{align}
\end{definition}
 \noindent Note that $\mathfrak{g}^{\perp}$ is not a Lie algebra.

\section{General Results for $\mathfrak{G}$-invariance}\label{section:general}

In this section we determine conditions leading to  models that are $\mathfrak{G}$-invariant by design. These results are stated in a general problem-agnostic way and will be applied to specific datasets in Secs.~\ref{section:examples} and \ref{section:discrete}.

\subsection{A single symmetry group}
Let us first consider the case when there is a single symmetry group $\mathfrak{G}$ associated with all the instances in the dataset. We aim at finding models  $h_{\thv}^{(k)}$ from Hypothesis Class $\HC_1$ that are $\mathfrak{G}$-invariant, i.e., models such that $h_{\thv}^{(k)}(V\rho V\ad)=h_{\thv}^{(k)}(\rho)$ for all $V\in\mathfrak{G}$ and choice of parameters $\thv$. Defining
\begin{equation}
   \widetilde{O}(\thv) = U\ad(\thv)O U(\thv)\,,
\end{equation}
so that $h_{\thv}^{(k)}(\rho)=\Tr[\rho^{\otimes k}\widetilde{O}(\thv)]$, we explicitly have
\begin{align}
    h_{\thv}^{(k)}(V \rho V\ad)&=\Tr[ V^{\otimes k}\rho^{\otimes k} (V\ad)^{\otimes k} \widetilde{O}(\thv)]\label{eq:for-invariance}\,.
\end{align}
Evidently, the model will be invariant under $\mathfrak{G}$ if $[\widetilde{O}(\thv),V^{\otimes k}]=0$ for all $V\in\mathfrak{G}$. 
Thus, the following proposition holds:
\begin{proposition}\label{prop:no-post-processing-symmetries}
Let $h_{\thv}^{(k)}\in\HC_1$ be a model in Hypothesis Class~\ref{def:hyp_class_1}, and $\mathfrak{G}$ be the symmetry group associated with the dataset. The model will be $\mathfrak{G}$-invariant if  $\widetilde{O}(\thv)\in \CC^{( k)}(\mathfrak{G})$.
\end{proposition}

\begin{proof}
The proof of this proposition follows from Definition~\ref{def:symmetries}. If $\widetilde{O}(\thv)$ belongs to the vector space of the $k$-th order symmetries $\CC^{( k)}(\mathfrak{G})$ then,  $[\widetilde{O}(\thv),V^{\otimes k}]=0$, and thus $h_{\thv}^{(k)}(V \rho V\ad)=h_{\thv}^{(k)}(\rho)$ for all $V\in\mathfrak{G}$. 
\end{proof}

Furthermore, as previously discussed,  we can guarantee that $\widetilde{O}(\thv)$ in Proposition~\ref{prop:no-post-processing-symmetries} can always be taken as a Hermitian operator and thus as an observable. 

Complementary to Proposition~\ref{prop:no-post-processing-symmetries} we now prescribe a second way of ensuring $\mathfrak{G}$-invariance of the model when $\mathfrak{G}$ forms a Lie group. This is achieved when the operator $\widetilde{O}(\thv)$ can be taken orthogonal to $ (V \rho V\ad)^{\otimes k} $ for all $V$ in $\mathfrak{G}$ and for all $\rho$ in $\SC$. We formalize this statement in the following proposition, proved in Appendix~\ref{app:B}.

\begin{proposition}\label{prop:no-post-processing-orthogonal}
Let $h_{\thv}^{(k)}\in\HC_1$ be a model in  Hypothesis Class~\ref{def:hyp_class_1}. Then, let $\mathfrak{G}$ be the symmetry Lie group associated with the dataset, and let $\mathfrak{g}\subseteq\mathfrak{u}(d)$ be its Lie algebra with $i\id\in\mathfrak{g}$. 
The model will be $\mathfrak{G}$-invariant when $\rho \in i \mathfrak{g}$ and $\widetilde{O}(\thv)\in{\rm span}(\{A_j\otimes A_{\overline{j}}\}_j) $. Here, $A_j\in i\mathfrak{g}^{\perp}$, an element of the orthogonal complement of $\mathfrak{g}$, is a Hermitian operator acting on the $j$-th copy of $\rho$, and $A_{\overline{j}}$ is an operator acting on all copies of $\rho$ but the $j$-th one.
\end{proposition}

Note that in Proposition~\ref{prop:no-post-processing-orthogonal} we have assumed that $i\id\in\mathfrak{g}$, where $\id$ denotes the $d\times d$ identity matrix. However, it could happen that $i\id$ is in $\mathfrak{g}^{\perp}$ instead. In this case, the proposition will hold if   $\rho \in i \mathfrak{g}^{\perp}$ and    $A_j\in i\mathfrak{g}$, as $\rho$ needs to have support on a vector space containing the identity.

\subsection{Multiple symmetry groups}
Let us now consider the case when each of the two classes in the dataset have a different symmetry group  associated to them, which we denote as $\mathfrak{G}_0$ and $\mathfrak{G}_1$.  The concepts used in the previous section to obtain group-invariant models, i.e., commutant and orthogonal complement, can also be leveraged to derive conditions under which a model $h_{\thv}^{(k)}$ is $\mathfrak{G}_0$-invariant,  $\mathfrak{G}_1$-invariant, or both. The following proposition, proved in Appendix~\ref{app:C}, generalizes Proposition~\ref{prop:no-post-processing-symmetries} to the case of two symmetry groups.
\begin{proposition}\label{prop:no-post-processing-symmetries-two}
Let $h_{\thv}^{(k)}\in\HC_1$ be a model in Hypothesis Class~\ref{def:hyp_class_1}, and let $\mathfrak{G}_0$  and $\mathfrak{G}_1$ be the symmetry groups associated with the dataset. 
The model will be $\mathfrak{G}_0$- and $\mathfrak{G}_1$-invariant if $\widetilde{O}(\thv)\in \CC^{( k)}(\mathfrak{G}_0)$ and $\widetilde{O}(\thv)\in \CC^{( k)}(\mathfrak{G}_1)$.
In addition, the model will be $\mathfrak{G}_{0}$-invariant but not necessarily  $\mathfrak{G}_{1}$-invariant if $\widetilde{O}(\thv)\in \CC^{( k)}(\mathfrak{G}_{0})$ but $\widetilde{O}(\thv)\not\in\CC^{( k)}(\mathfrak{G}_{1})$.
\end{proposition}

Conversely, while not stated explicitly in Proposition~\ref{prop:no-post-processing-symmetries-two}, the model will be $\mathfrak{G}_{1}$-invariant but not necessarily  $\mathfrak{G}_{0}$-invariant if $\widetilde{O}(\thv)\in \CC^{( k)}(\mathfrak{G}_{1})$ but $\widetilde{O}(\thv)\not\in\CC^{( k)}(\mathfrak{G}_{0})$.

Additionally, when $\mathfrak{G}_{0}$ and $\mathfrak{G}_{1}$ are Lie groups, with associated Lie algebras $\mathfrak{g}_{0}$ and $\mathfrak{g}_{1}$,  we can generalize Proposition~\ref{prop:no-post-processing-orthogonal} to the case of two symmetry groups. Then, the following proposition, proved in Appendix~\ref{app:C}, holds.
\begin{proposition}\label{prop:no-post-processing-orthogonal-two}
Let $h_{\thv}^{(k)}\in\HC_1$ be a model in Hypothesis Class~\ref{def:hyp_class_1}, and let $\mathfrak{G}_0$  and $\mathfrak{G}_1$ be the symmetry Lie groups associated with the dataset, with $\mathfrak{g}_0$ and $\mathfrak{g}_1$ their respective Lie algebras with $i\id\in\mathfrak{g}_0,\mathfrak{g}_1$. The model will be $\mathfrak{G}_0$-invariant and $\mathfrak{G}_1$-invariant when $\rho \in i \mathfrak{g}_0,i \mathfrak{g}_1$ and when  $\widetilde{O}(\thv)\in{\rm span}(\{A_j\otimes A_{\overline{j}}\}_j) $. Here, $A_j\in i\mathfrak{g}_0^{\perp},i\mathfrak{g}_1^{\perp}$ is a Hermitian operator acting on the $j$-th copy of $\rho$, $A_{\overline{j}}$ is an operator acting on all copies of $\rho$ but the $j$-th one . In addition, the model will be $\mathfrak{G}_{0(1)}$-invariant but not necessarily  $\mathfrak{G}_{1(0)}$-invariant when $\rho \in i \mathfrak{g}_0,i \mathfrak{g}_i$  and $A_j\in i\mathfrak{g}_0^{\perp}$ but $A_j\not\in i\mathfrak{g}_1^{\perp}$. 
\end{proposition}
In Proposition~\ref{prop:no-post-processing-orthogonal-two} we have assumed that $i\id$ belongs to $\mathfrak{g}_0$ and $\mathfrak{g}_1$. However, if $i\id$ instead belongs to $\mathfrak{g}^{\perp}_i$ (with $i=0,1$), then the proposition will hold by replacing $\mathfrak{g}_i $ by $\mathfrak{g}^{\perp}_i$, and conversely. 

Propositions~\ref{prop:no-post-processing-symmetries}--\ref{prop:no-post-processing-orthogonal-two} provide conditions under which one can guarantee that a QML model in Hypothesis Class~\ref{def:hyp_class_1} is $\mathfrak{G}$-invariant. While the results presented in this section are valid for the case when there are two symmetry groups, the previous propositions can be readily extended to more general scenarios (such as multi-class classification), where one has a set $\{\mathfrak{G}_i\}_i$ of symmetry groups. For instance, one could generalize Proposition~\ref{prop:no-post-processing-symmetries-two} to show that a model will be invariant under all symmetry groups $\mathfrak{G}_i$ if $\widetilde{O}(\thv)\in \CC^{( k)}(\mathfrak{G}_i)$ for all $i$.

\section{Lie Group-Invariant Models}\label{section:examples}

We now apply the general results presented in the previous section to identify $\mathfrak{G}$-invariant models that can classify states originating from several paradigmatic quantum datasets whose invariances are captured by  Lie groups. These include the  \textit{purity dataset} (Sec.~\ref{sec:purity}), the \textit{time-reversal dataset} (Sec.~\ref{sec:time-reversal}), and the \textit{multipartite entanglement dataset} (Sec.~\ref{sec:entanglement}). Our results are stated in the form of theorems.  For pedagogical reasons, we include in the main text the proofs for most of these theorems, as they provide a constructive introduction to our framework.

\subsection{Purity dataset}
\label{sec:purity}

As a first application, we consider the  QML task of classifying $n$-qubit states according to their purity~\cite{huang2021quantum}. Given a dataset $\SC=\{(\rho_i,y_i)\}_{i=1}^N$, we want to discriminate those $\rho_i$ that are pure from those that are not. That is, we assign labels 
\begin{equation}
    y_i = 
    \begin{cases}
        1 & \text{if } \Tr[\rho_i^2]=1, \\
        0 & \text{if } \Tr[\rho_i^2]=b<1\,
    \end{cases}
\end{equation}
to states $\rho_i$ according to values of their purities $\Tr[\rho_i^2]$.

The symmetry group $\mathfrak{G}$ associated with the data in \textit{both} classes is the group of unitaries $\mathbb{U}(d)$. This follows from the fact that unitaries preserve the spectral properties of quantum states, and thus their purity remain unchanged under the action of $\mathbb{U}(d)$.

\subsubsection{Conventional experiments}
Let us first consider the case of conventional experiments (see Definition~\ref{def:conventional_experiment}), i.e., when the model $h_{\thv}^{(1)}$ in Eq.~\eqref{eq:linear_cost} has access to $k=1$ copy of each data at a time. For such a case, we can derive the following theorem.

\begin{theorem}\label{theo:purity-conventional}
Let $h_{\thv}^{(1)}\in\HC_1$ be a model in Hypothesis Class~\ref{def:hyp_class_1}, computable in a conventional experiment. There exists no quantum neural network $U(\thv)$ and operator $O$ such that $h_{\thv}^{(1)}$ is invariant under the action of $\mathbb{U}(d)$ and can classify (i.e., provide any relevant information about) the data in the purity dataset. 
\end{theorem}

\begin{proof}
The strategy of this proof is as follows. First, we identify the possible $\mathfrak{G}$-invariant models arising from Propositions~\ref{prop:no-post-processing-symmetries} and~\ref{prop:no-post-processing-orthogonal}. 
Then, we show that these models cannot be used to perform classification 
for the purity dataset. 
Finally, we prove that no other $\mathfrak{G}$-invariant model within $\HC_1$ (with $k=1$ copies) exist.

Recall from  Proposition~\ref{prop:no-post-processing-symmetries} that a model is $\mathfrak{G}$-invariant under $\mathbb{U}(d)$ if  $\widetilde{O}(\thv)$ is in the commutant of $\mathbb{U}(d)$. 
Since $\mathfrak{G}=\mathbb{U}(d)$ is irreducible~\footnote{A representation is irreducible if it cannot be further decomposed into a direct sum of representations.} in $\mathbb{C}^{d \times d}$, we know from Schur's Lemma~\cite{kirillov2008introduction} that 
\begin{equation}
    \CC(\mathfrak{G})={\rm span}(\{\id\})\,.
\end{equation}
 It follows that if $\widetilde{O}(\thv)$ is in $\CC(\mathfrak{G})$, it takes the form $\widetilde{O}(\thv)=\lambda \id$. Moreover, we impose $\lambda \in\mathbb{R}$  to ensure the Hermiticity of $\widetilde{O}(\thv)$. This yields a constant model prediction $h_{\thv}^{(1)}(\rho)=\lambda$ for any $\rho$. 
Hence, the $\mathfrak{G}$-invariant models of Proposition~\ref{prop:no-post-processing-symmetries} do not provide any information about the purity of a state, and thus cannot classify the data.  

Let us now analyze the models arising from Proposition~\ref{prop:no-post-processing-orthogonal}. Since $\mathfrak{g}=\mathfrak{u}(d)$, we have 
\begin{equation}\label{eq:perp:unitary}
\mathfrak{g}^\perp=\{\vec{0}\}\,,
\end{equation}
where $\vec{0}$ denotes the $d \times d$ null  matrix. Hence, if $\widetilde{O}(\thv)\in i\mathfrak{g}^\perp$, then $h_{\thv}^{(1)}(\rho)=0$ for any $\rho$. This shows that the  $\mathfrak{G}$-invariant models arising from  Proposition~\ref{prop:no-post-processing-orthogonal} do not provide any information about the purity of a state and cannot classify the data. 

So far, we have seen that $\mathfrak{G}$-invariant models obtained by applying Propositions~\ref{prop:no-post-processing-symmetries} and~\ref{prop:no-post-processing-orthogonal} do not allow for classification of the purity dataset. Still, this does not preclude the existence of other $\mathfrak{G}$-invariant models within  $\HC_1$ that may be adequate for classification. 
However, we now prove that no other $\mathfrak{G}$-invariant models exist, beyond those already considered. 
Given that $h_{\thv}^{(1)}(V \rho V\ad)=h_{\thv}^{(1)}(\rho)$ should be true for any unitary $V$ in $\mathbb{U}(d)$, the latter also needs to hold when uniformly averaging $h_{\thv}^{(1)}(V \rho V\ad)$ over all possible $V$ in $\mathfrak{G}$. Namely, we require that $\mathbb{E}_{\mathfrak{G}}[h_{\thv}^{(1)}(V \rho V\ad)]=h_{\thv}^{(1)}(\rho)$. The left-hand-side of the equality is evaluated as
\begin{align}
     \mathbb{E}_{\mathfrak{G}}[h_{\thv}^{(1)}(V \rho V\ad)]&=\int_{\mathbb{U}(d)}d\mu(V)\Tr[U(\thv)V\rho V\ad\ U\ad(\thv)O]\nonumber\\
    &=\Tr[\left(\int_{\mathbb{U}(d)}d\mu(V) U(\thv)V\rho (U(\thv)V)\ad\right) O]\nonumber\\
    &=\Tr[\left(\int_{\mathbb{U}(d)}d\mu(V) V\rho V\ad\right) O]\nonumber\\
    &=\frac{\Tr[\rho]\Tr[O]}{d}\,,\label{eq:no_c_1}
\end{align}
where the integral denotes the Haar average over the unitary group. 
In the second equality, we have used the linearity of the trace and of the integral. Then, in the third equality, we have used the left-invariance of the Haar-measure. Finally, in the fourth equality, we have explicitly computed the integration via the Weingarten calculus~\cite{collins2006integration,puchala2017symbolic}. 
From Eq.~\eqref{eq:no_c_1} we can see that the only way for  $h_{\thv}^{(1)}(\rho)$ to be equal to $\Tr[\rho]\Tr[O]$ for general quantum states is to have $O\propto\id/d$ or $O=\vec{0}$, which leads to the solutions given by Propositions~\ref{prop:no-post-processing-symmetries} and~\ref{prop:no-post-processing-orthogonal}. Hence, we have shown that there are no models in the Hypothesis Class~\ref{def:hyp_class_1} that are $\mathfrak{G}$-invariant under the action of $\mathbb{U}(d)$  and that can classify the data in the purity dataset as they all provide no information according to Definition~\ref{def:accuracy}.
\end{proof}

The previous proposition shows that one cannot classify the data in the purity dataset with a model in $\HC_1$ operating in a conventional experiment. In hindsight, one could have foreseen this result.  Indeed, computing the purity requires evaluating a polynomial of order two in the matrix elements of $\rho$, and thus, the linear functions as the ones here considered are deemed to fail. 
From a QML perspective, $h_{\thv}^{(1)}$ is ultimately a linear classifier where the parameterized quantum neural network $U(\thv)$ defines a hyperplane such that the expectation value of $\widetilde{O}(\thv)$ is positive for one class and negative for the other. However, the manifolds of quantum states with different purities are not linearly separable in the state space. This can be better exemplified by single-qubit states in the Bloch sphere, where no plane can be drawn across the sphere which linearly separates pure and mixed states. 

In the spirit of kernel tricks~\cite{cristianini2000introduction}, one can introduce non-linearities by allowing the models  $h_{\thv}^{(k)}$ to coherently access multiple copies of $\rho$. This is precisely the setting of quantum-enhanced experiments, which we now explore.

\subsubsection{Quantum-enhanced experiments}

We now consider the case of quantum-enhanced experiments (see Definition~\ref{def:quantum_enhanced}) where multiple copies of a state in the dataset can be operated over in a coherent manner. As we now see, $k=2$ copies is already enough for classifying states according to their purity.
\begin{theorem}\label{theo:unitary-quadratic}
Let $h_{\thv}^{(2)}\in\HC_1$ be a model in Hypothesis Class~\ref{def:hyp_class_1}, computable in a quantum-enhanced experiment.  There always exists quantum neural networks $U(\thv)$ and operators $O$, resulting in $\widetilde{O}(\thv)\in{\rm span}(\{\id \otimes \id,SWAP\})$, such that $h_{\thv}^{(2)}$ is invariant under the action of $\mathbb{U}(d)$. If the model has non-zero component in $SWAP$, it can perfectly classify the data in the purity dataset. The special choice of $\widetilde{O}(\thv)=SWAP$ leads to $h_{\thv}^{(2)}(\rho)=\Tr[\rho^2]$. 
\end{theorem}

\begin{proof}
Recall from  Proposition~\ref{prop:no-post-processing-symmetries} that a model $h_{\thv}^{(2)}$ in the Hypothesis Class~\ref{def:hyp_class_1} is $\mathfrak{G}$-invariant if $\widetilde{O}(\thv)$ is a quadratic symmetry of $\mathbb{U}(d)$, i.e., if for all $V\in \mathbb{U}(d)$
\begin{equation}\label{eq:explicit-quadratic}
  (V\ad)^{\otimes 2}\widetilde{O}(\thv)V^{\otimes 2}=\widetilde{O}(\thv)\,.
\end{equation}

\begin{figure}[t]
\centering
\includegraphics[width=1\columnwidth]{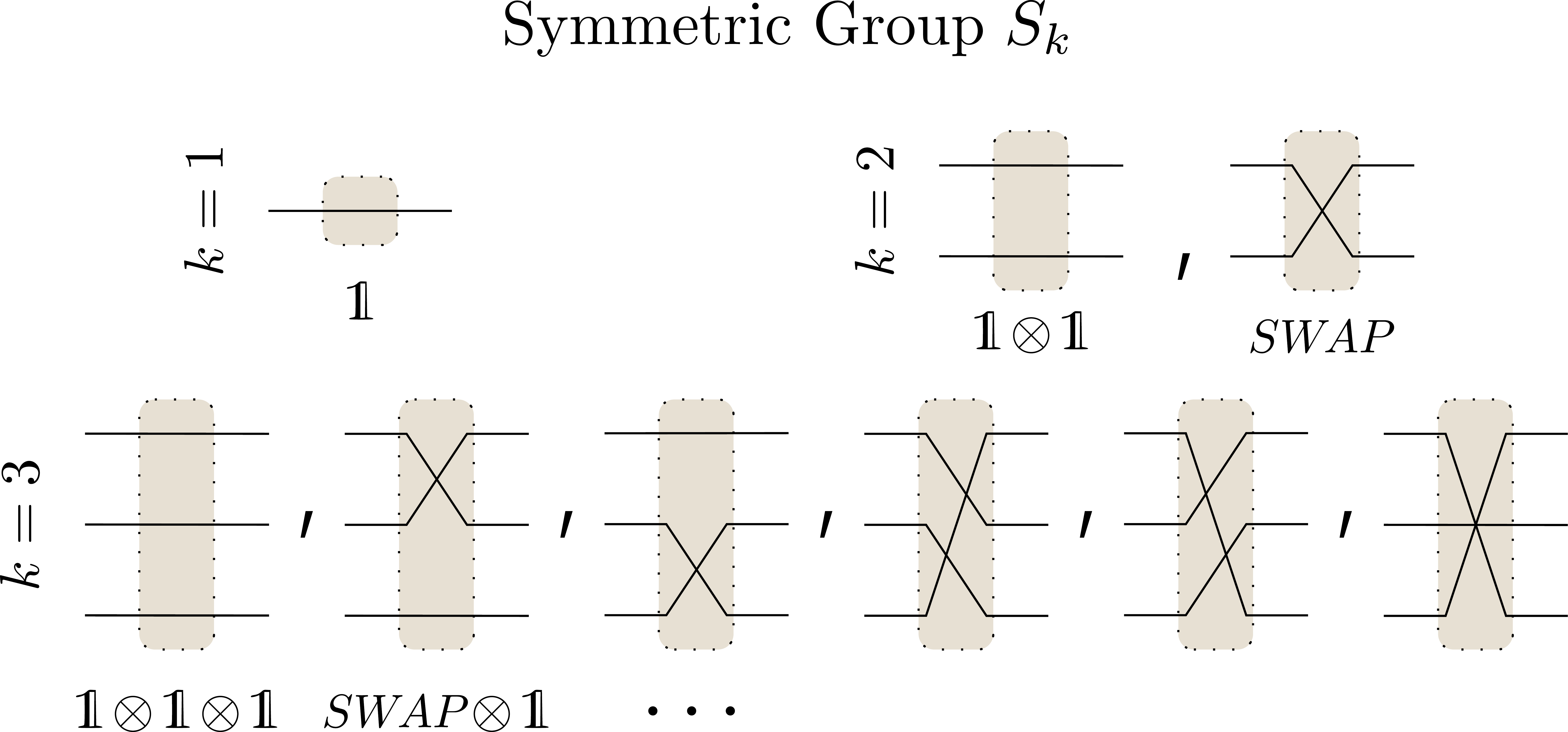}
\caption{\textbf{Elements of the Symmetric group.} 
The Symmetric group $S_k$ is composed of the $k!$ distinct permutations over $k$ indices. Here we illustrate its representation acting on tensor product systems for the cases of $k=1,2$ and $3$ copies of n-qubit states (each represented as a line). For example, the element $SWAP\otimes \id\in S_3$ depicted second from left-to-right, acts on a tensor product state as $SWAP\otimes \id\ket{\psi_1}\ket{\psi_2}\ket{\psi_3}=\ket{\psi_2}\ket{\psi_1}\ket{\psi_3}$.}
\label{fig:quadratic_symmetries-3}
\end{figure}

From the Schur-Weyl duality~\cite{goodman2009symmetry} we know that the $k$-th order symmetries of $\mathbb{U}(d)$ are given by \begin{equation}\label{eq:k_symmetries_unitary}
    \CC^{( k)}(\mathbb{U}(d))={\rm span}\,(S_k)\,,
\end{equation}
with $S_k$ the representation of the Symmetric Group that acts by permuting subsystems of the $k$-fold tensor product of the input state (depicted in Fig.~\ref{fig:quadratic_symmetries-3}).

As shown in Fig.~\ref{fig:quadratic_symmetries-3}, for the case of $k=2$ copies, this group contains only two elements~\footnote{We refer the reader to~\cite{zeier2011symmetry,zimboras2015symmetry} for an in-depth discussion on the quadratic symmetries of $\mathbb{U}(d)$.}
\begin{equation}
 S_2=\{\id\otimes\id,SWAP\}\,,
\end{equation}
with $\id\otimes\id$ the identity acting on each of the two copies of $\rho$, and $SWAP$ the operator swapping these copies. 
As a consequence, $h_{\thv}^{(2)}$ can be made $\mathfrak{G}$-invariant under the action of $\mathbb{U}(d)$ when $\widetilde{O}(\thv)=a_1\,\id\otimes\id + a_2\, SWAP$ with $a_1,a_2\in\mathbb{R}$. The latter is diagrammatically presented in Fig.~\ref{fig:quadratic_symmetries}. This yields predictions
\begin{equation}
    h_{\thv}^{(2)}(\rho)=a_1 +a_2 \Tr[\rho^2]\,,
\end{equation}
showing that the model will be able to perfectly classify the states in the purity dataset (according to Definition~\ref{def:accuracy}) for any choice of $a_2\neq 0$. 
\end{proof}

\begin{figure}[t]
\centering
\includegraphics[width=1\columnwidth]{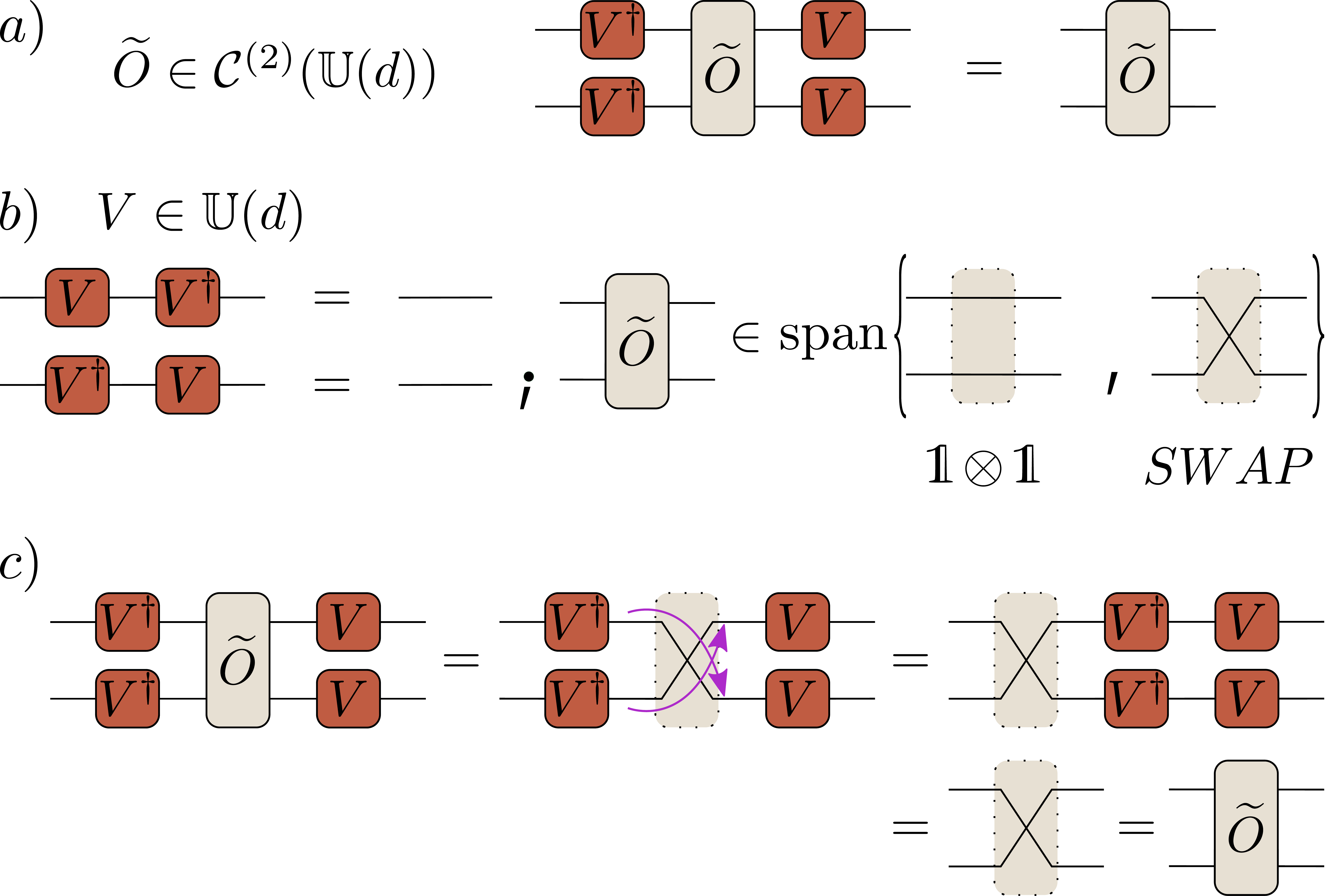}
\caption{\textbf{Quadratic symmetries for $\mathfrak{G}=\mathbb{U}(d)$.} a) A model $h_{\thv}^{(2)}$ is $\mathfrak{G}$-invariant when $\widetilde{O}(\thv)$ is a quadratic symmetry of $\mathfrak{G}$, that is, when $(V\ad)^{\otimes 2}\widetilde{O}(\thv)V^{\otimes 2}=\widetilde{O}(\thv)$ for all $V\in\mathfrak{G}$.
This equality is shown in diagrammatic tensor representation, where each line corresponds to a $d$-dimensional Hilbert space hosting a copy of $\rho$, and boxes represent unitary operations. b) For $V\in\mathbb{U}(d)$, we have $V V\ad =V\ad V =\id$ depicted on the left part. Furthermore, we know that the quadratic symmetries of $\mathbb{U}(d)$ are  spanned by the identity $\id\otimes\id$ and the SWAP operators. We depict these on the right. c) Using the diagrammatic tensor representation we verify that $ (V\ad)^{\otimes 2}\widetilde{O}(\thv)V^{\otimes 2}=\widetilde{O}(\thv)$, for the case when $\widetilde{O}(\thv)=SWAP$. }
\label{fig:quadratic_symmetries}
\end{figure}

We now make  several remarks regarding the results in Theorem~\ref{theo:unitary-quadratic}, and regarding our framework in general. First, we note that while Proposition~\ref{prop:no-post-processing-symmetries} provides a straightforward guideline to obtain $\mathfrak{G}$-invariant models, it does not prescribe how to actually build the quantum neural networks $U(\thv)$ and the measurement operators $O$ ensuring that $\widetilde{O}(\thv)=U\ad(\thv) O U(\thv)$ is a $k$-th order symmetry. All that we know is the specific form that the resulting $\widetilde{O}(\thv)$ needs to have. 
Thus, it is still necessary to find an adequate ansatz for $U(\thv)$ and an appropriate observable $O$ that can be efficiently measured. For instance, it is clear that simply choosing $U(\thv)=\id\otimes \id$ $\forall \thv$ and $O=SWAP$ satisfies the conditions of Theorem~\ref{theo:unitary-quadratic}. This has the issue that one cannot efficiently estimate the expectation value of the SWAP operator -- its Pauli decomposition has a number of terms that scales exponentially with $n$ -- using a model such as $h^{(2)}_{\thv}$. However, it is well known that by adding an ancilla qubit and by using the Hadamard-test~\cite{buhrman2001quantum} one can efficiently estimate the expectation value of the SWAP operator. In Appendix~\ref{app:D}, we show how our present formalism can be applied to models that include an ancillary qubit. Surprisingly, the latter  allows us to discover a new connection between the Swap Test~\cite{buhrman2001quantum} and the ancilla based algorithm of Ref.~\cite{cincio2018learning}.

\subsection{Time-reversal dataset}
\label{sec:time-reversal}
In this section we  are  interested in classifying states according to whether they are obtained from a time-reversal-symmetric~\cite{sachs1987physics} dynamic or from an arbitrary one. 
That is, the states $\rho_i$ of the corresponding dataset $\SC=\{(\rho_i,y_i)\}_{i=1}^N$ now have labels
\begin{equation}
    y_i = 
    \begin{cases}
        1 & \text{if $\rho_i$ is real valued,}\\
        0 & \text{if $\rho_i$ is Haar random.
        }
    \end{cases}
\end{equation}
Specifically, the states in the dataset have a label $y_i=1$ if they are generated by evolving some (fixed) \textit{real-valued} fiduciary  state with a time-reversal-symmetric unitary (and thus are real-valued too), and a label $y_i=0$  if they are generated by evolving the same reference state with a  Haar random unitary.

In contrast to the case of the purity dataset previously considered, one can now associate a distinct symmetry group to each of the two classes.
On one hand, the states with label $y_i=0$ have $\mathfrak{G}_0=\mathbb{U}(d)$ as a symmetry group. On the other hand, the states with label $y_i=1$ have, as a symmetry group, $\mathfrak{G}_1=\mathbb{O}(d)$, which is the orthogonal Lie group of degree $d$. This is because the unitaries in $\mathbb{O}(d)$ preserve the time-reversal symmetry of the states (and thus their label). 

For convenience, we recall that  $\mathbb{O}(d)$ is the group of $d\times d$ orthogonal matrices. That is, $\forall V\in \mathbb{O}(d)$, $VV^t=V^t V=\id$. This group can be obtained by exponentiation of the orthogonal Lie algebra, which consists of $d \times d$ skew-symmetric matrices, $\mathfrak{g}_1=\mathfrak{o}(d)$, i.e., $\mathbb{O}(d)=e^{\mathfrak{o}(d)}$. 
Moreover, the unitary Lie algebra $\mathfrak{g}_0=\mathfrak{u}(d)$ can be split as $    \mathfrak{u}(d)=\mathfrak{o}(d)\oplus \mathfrak{u}_{\mathbb{C}}(d)$. 
Here, note that $\mathfrak{o}(d)$  corresponds to the purely real-valued subspace of the unitary algebra. Its orthogonal complement
\begin{equation}\label{eq:perp-orthogonal}
\mathfrak{g}_1^\perp=\mathfrak{u}_{\mathbb{C}}(d) 
\end{equation}
corresponds to the purely imaginary subspace of the unitary Lie algebra. 

Having two symmetry classes allows for the design of a new classification strategy.
Namely,  one can classify the data using a $\mathfrak{G}_1$-invariant model (but not $\mathfrak{G}_0$-invariant) $h_{\thv}^{(k)}$ such that
\begin{equation}\label{eq:classification_reversal}
\begin{split}
    &h_{\thv}^{(k)}(\rho_i)=
    c\quad \text{if $y_i=1$}\,,\\
    &h_{\thv}^{(k)}(\rho_i)\in[b_1,b_2] \quad \text{if $y_i=0$}\,,
\end{split}
\end{equation}
with $c$, $b_1$ and $b_2$ real values determined by the measurement operator and the states in the dataset. If $c\not\in [b_1,b_2]$, then Eq.~\eqref{eq:classification_reversal} suffices for perfect classification according to Definition~\ref{def:accuracy}. If $c\in[b_1,b_2]$, we can still use Eq.~\eqref{eq:classification_reversal} for noisy classification (see Definition~\ref{def:accuracy}) but there could be a chance of misclassification  as one cannot perfectly distinguish  between states in different classes yielding the same prediction. Such misclassification events will remain unlikely as long as the probability that $h_{\thv}^{(k)}(\rho_i)=c$ (up to some additive error) is small for states with label $y_i=0$. In Appendix~\ref{app:E}, we present a Lemma that formalizes the previous statement. In any case, for now we assume that a model satisfying Eq.~\eqref{eq:classification_reversal} can classify the data in the dataset with probability high enough, and will challenge this assumption in due course.

\subsubsection{Conventional experiments}
For the case of conventional experiments, i.e., $k=1$ copies in Eq.~\eqref{eq:linear_cost}, the results in Proposition~\ref{prop:no-post-processing-symmetries-two} cannot be used to find $\mathfrak{G}_1$-invariant models classifying the time-reversal dataset. 
Indeed, since the representation of $\mathfrak{G}_1=\mathbb{O}(d)$ is irreducible, using Schur's Lemma~\cite{kirillov2008introduction} we know that 
\begin{equation}
    \CC(\mathfrak{G}_1)={\rm span}(\{\id\})\,,
\end{equation}
i.e., $\mathfrak{G}_1$ has no non-trivial linear symmetries that could be exploited for the purpose of classification. 

Still, we can use Proposition~\ref{prop:no-post-processing-orthogonal-two} to find group invariant models. First, we note that the input states $\rho_i$ belong to $\mathfrak{g}_1^{\perp}=\mathfrak{u
}_{\mathbb{C}}(d)$ when they are time-reversal-symmetric but to $\mathfrak{g}_0=\mathfrak{u}(d)$ when they are Haar random. 
Hence,  $h_{\thv}^{(1)}$ will be invariant under the action of $\mathfrak{G}_1$, but not necessarily invariant under the action of $\mathfrak{G}_0$, if $\widetilde{O}(\thv)$ is in $i\mathfrak{g}_{1}$ but not in $i\mathfrak{g}_{0}^{\perp}$. This is formalized below.

\begin{theorem}\label{th:time-reversal-conventional}
Let $h_{\thv}^{(1)}\in\HC_1$ be a model in Hypothesis Class~\ref{def:hyp_class_1}, computable in a conventional experiment. There always exist real-valued quantum neural networks $U(\thv)$ and operators $O$, resulting in $\widetilde{O}(\thv)\in i\mathfrak{o}(d)$, with $\widetilde{O}(\thv)\neq\vec{0}$, such that $h_{\thv}^{(1)}$ is invariant under the action of $\mathbb{O}(d)$  and can perform noisy classification of the data in the time-reversal dataset.
\end{theorem}

\begin{proof}
We aim at finding models that are $\mathfrak{G}_1$-invariant (with $\mathfrak{G}_1=\mathbb{O}(d))$ but not necessarily $\mathfrak{G}_0$-invariant (with $\mathfrak{G}_0=\mathbb{U}(d))$, distinguishing time-reversal-symmetric states from Haar random ones. According to Proposition~\ref{prop:no-post-processing-orthogonal-two}, the model will be $\mathfrak{G}_1$-invariant but not $\mathfrak{G}_0$-invariant if $\widetilde{O}(\thv)\in i\mathfrak{g}_1=i\mathfrak{o}(d)$ and $\widetilde{O}(\thv)\neq\vec{0}$, e.g., if  $\widetilde{O}(\thv)$ is a purely imaginary operator.

Now, lets show that there is a choice of $O$ and $U(\thv)$ allowing for classification. Taking $O\in i\mathfrak{g}_1$ and $U(\thv)\in\mathfrak{G}_1=\mathbb{O}(d)$, the resulting $\widetilde{O}(\thv)$ is also contained in $i\mathfrak{g}_1$, since a Lie algebra is closed under the action of its associated Lie group. Because  time-reversal-symmetric states $\rho_i$ are exclusively contained in $i\mathfrak{g}_1^{\perp}$, it follows from Eq.~\eqref{eq:orthogonal_complement} that
\begin{equation}
    h_{\thv}^{(1)}(\rho_i)=0, \quad \forall \rho_i \text{ with label } y_i=1\,.\nonumber
\end{equation}
Moreover, the previous equation is not satisfied for Haar random states, as these will generally have both real and complex parts. As such, $h_{\thv}^{(1)}(\rho_i)$ will not necessarily be zero for states with label  $y_i=0$. Hence, the model satisfies Eq.~\eqref{eq:classification_reversal} such that it can perform noisy classification (according to Definition~\ref{def:accuracy}) for the states in the time-reversal dataset.
\end{proof}

So far, we have identified models that yield predicted values of $0$ for time-reversal-symmetric states, but yield values in a continuous range for states drawn from the Haar distribution. As such, when taking into account noise in the prediction of the model, any non-time-reversal state with prediction values close to zero may be misclassified. 
In fact, as proven in Appendix~\ref{app:F}, Haar random states lead to prediction values that (with probability close to one) lie in a range that becomes exponentially concentrated around zero with the number of qubits $n$. 
In turn, it can be shown that to classify states in the dataset with a success probability of at least $2/3$, one would need to repeat the experiment a number of times that scales as $\Omega(2^{2n/7})$~\cite{chen2021exponential,aharonov2022quantum,huang2021quantum}.

This raises attention towards a practical aspect in the design of QML models that we have not previously considered: the \textit{scaling} in the number of experiment repetitions required for accurate classification. Our framework allows us to identify $\mathfrak{G}$-invariant models, but we are not guaranteed that such models are practical for large  system sizes $n$. In fact, we have seen that an exponential number of repetitions are needed to make practical use of the  models in Theorem~\ref{th:time-reversal-conventional}. This motivates us to further continue the search of $\mathfrak{G}$-invariant models in quantum-enhanced experiments in the hope that these might avoid the exponential scaling present in conventional experiments. 

\subsubsection{Quantum-enhanced experiments}
For quantum-enhanced experiments, i.e., $k=2$ copies in Eq.~\eqref{eq:linear_cost}, we can show that the following theorem holds.
\begin{theorem}\label{theo:orthogonal-quadratic}
    Let $h_{\thv}^{(2)}\in\HC_1$ be a model in Hypothesis Class~\ref{def:hyp_class_1}, computable in a quantum-enhanced experiment.  There always exist quantum neural networks $U(\thv)$ and operators $O$, resulting in  $\widetilde{O}(\thv)=\dya{\Phi^+}$, with $\ket{\Phi^+}$ the Bell state on $2n$ qubits, such that $h_{\thv}^{(2)}$ is invariant under the action of $\mathbb{O}(d)$ and can perform noisy classification of the data in the time-reversal dataset. 
\end{theorem}

\begin{figure}[t]
\centering
\includegraphics[width=1\columnwidth]{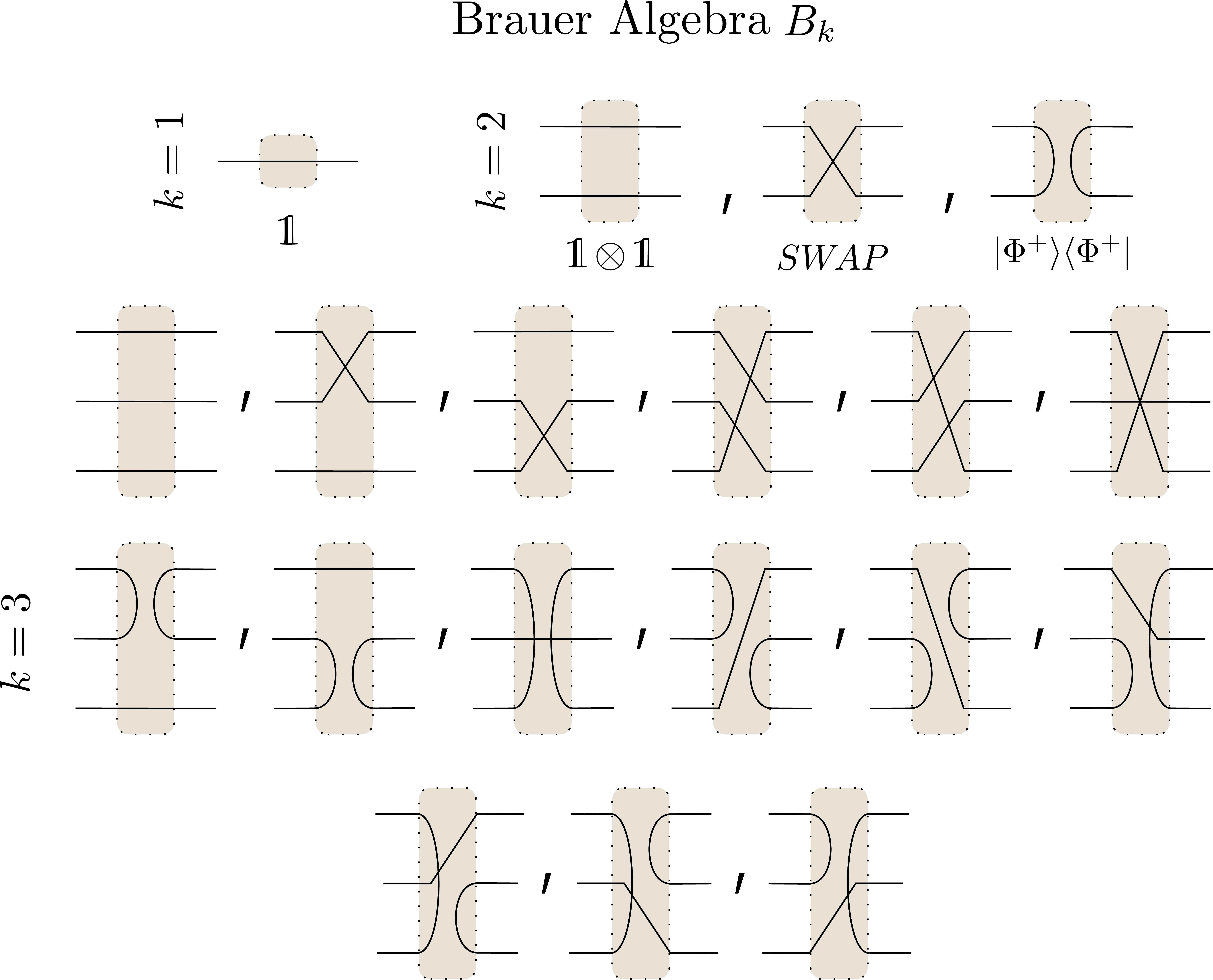}
\caption{\textbf{Elements of the Brauer algebra.} A basis for the Brauer algebra $B_k$ is composed of $2k!/(2^kk!)$ possible pairings on a set of $2k$ elements, where any element may be matched to another. Here we illustrate its representation acting on tensor product systems for the cases of $k=1,2$ and $3$ copies. }
\label{fig:quadratic_symmetries-3-O}
\end{figure}

\begin{proof}
We aim at finding models that are invariant under $\mathbb{O}(d)$, but not under $\mathbb{U}(d)$. According to Proposition~\ref{prop:no-post-processing-symmetries-two}, this can be achieved by ensuring that $\widetilde{O}(\thv)$ is a quadratic symmetry of $\mathbb{O}(d)$ but not of $\mathbb{U}(d)$. From  the Schur-Weyl duality we know that  the $k$-th order symmetries of $\mathbb{O}(d)$ are given by the Brauer algebra $B_k$~\cite{brown1954algebra},
\begin{equation}\label{eq:k_symmetries-orthogonal}
    \CC^{( k)}(\mathbb{O}(d))=B_k\,.
\end{equation}
The elements of $B_k$ are depicted in Fig.~\ref{fig:quadratic_symmetries-3-O}.

As shown in Fig.~\ref{fig:quadratic_symmetries-O}, for $k=2$ the Brauer algebra is spanned by three elements
\begin{equation}\label{eq:brauer-2}
    B_2={\rm span}(\left\{\id\otimes\id,SWAP,\dya{\Phi^+}\right\})\,,
\end{equation}
where $\ket{\Phi^+}$ denotes the Bell state on $2n$ qubits
\begin{equation}
    \ket{\Phi^+}=\frac{1}{d}\sum_{j=1}^d\ket{j}\ket{j}\,.
\end{equation}
It can be verified that $\dya{\Phi^+}$ is indeed a quadratic symmetry for $\mathbb{O}(d)$. To see that, recall the \textit{ricochet property} (also called the \textit{transpose trick}), which states that for any linear operator $A$ acting on a $d$-dimensional Hilbert space
\begin{equation}\label{eq:ricochet}
    (A\otimes \id)\ket{\Phi^+}=(\id\otimes A^t)\ket{\Phi^+}\,.
\end{equation}  
Using Eq.~\eqref{eq:ricochet} one can assert that $ (V^t)^{\otimes 2}\dya{\Phi^+} V^{\otimes 2}=(\id\otimes V^tV)\dya{\Phi^+}(\id\otimes V^t V)=\dya{\Phi^+}$, and hence that $\dya{\Phi^+}\in\CC^{(2)}(\mathbb{O}(d))$ (see also Fig.~\ref{fig:quadratic_symmetries-O}(b) for a diagrammatic proof).

\begin{figure}[t]
\centering
\includegraphics[width=1\columnwidth]{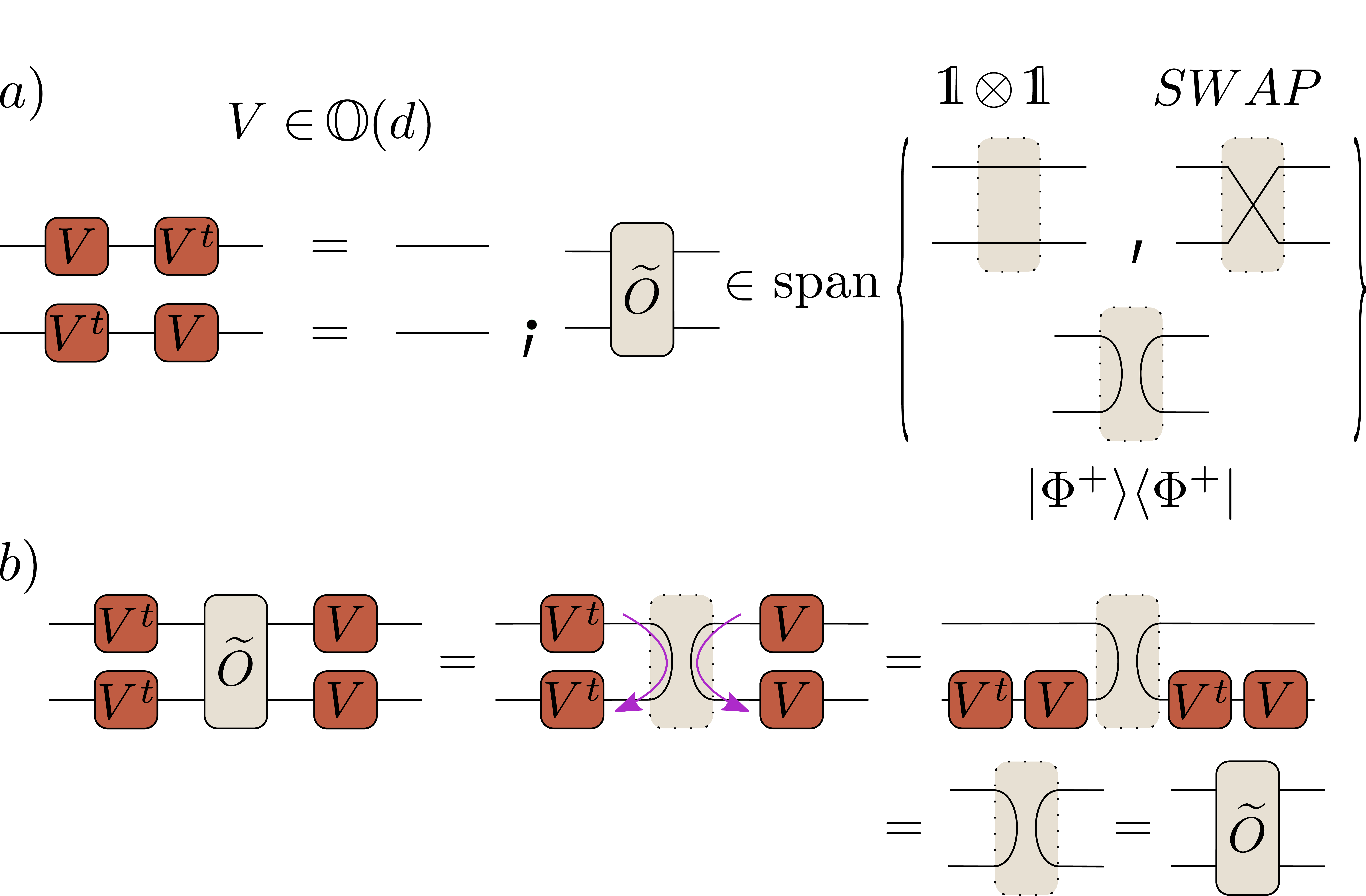}
\caption{\textbf{Quadratic symmetries for $\mathfrak{G}=\mathbb{O}(d)$.} a) We know that any $V\in\mathbb{O}(d)$ is such that $V V^t =V^t V =\id$. We schematically show this fundamental property on the left. We know that the quadratic symmetries of $\mathbb{O}(d)$ are elements of the Brauer algebra $B_2$ whose basis contains three elements: the identity $\id\otimes\id$, the SWAP operator, and the projector onto the Bell state $\dya{\phi^+}$. We depict these on the right. c) Using the diagrammatic tensor representation we verify that $ (V^t)^{\otimes 2}\widetilde{O}(\thv)V^{\otimes 2}=\widetilde{O}(\thv)$. For the  case of $\dya{\phi^+}$ we use the ricochet property of Eq.~\eqref{eq:ricochet}.}
\label{fig:quadratic_symmetries-O}
\end{figure}

The only element that is in $\CC^{(2)}(\mathbb{O}(d))$ but not in $\CC^{(2)}(\mathbb{U}(d))$ is the projector onto the Bell state $\dya{\Phi^+}$. Hence, according to Proposition~\ref{prop:no-post-processing-symmetries-two}, $h_{\thv}^{(2)}$ is $\mathbb{O}(d)$-invariant if $\widetilde{O}(\thv)=\lambda \dya{\Phi^+}$ with $\lambda\in\mathbb{R}$. Now, the model is such that  $h_{\thv}^{(2)}(\rho_i)=\lambda \bra{\Phi^+}(\rho_i^{\otimes 2})\ket{\Phi^+}$. In Fig.~\ref{fig:orthogonal-circuits}(a) we show a circuit that could be used to measure this overlap.

Recall that the time-reversal states $\rho_i$ are obtained by evolving a real-valued fiduciary state -- taken to be $\ket{0}^{\otimes n}$ without loss of generality -- under a unitary in $\mathbb{O}(d)$. One can verify that if $\widetilde{O}(\thv)=\dya{\Phi^+}$, then 
\begin{align}
    h_{\thv}^{(2)}(\rho_i)
    &=|\langle \Phi^+ | 0\rangle ^{\otimes 2n}|^2=\frac{1}{d^2}\,,
\end{align}
for all $\rho_i$ with labels $y_i=1$. On the other hand, the model output will not be constant for states with labels $y_i=0$, i.e., for states obtained  by evolving $\ket{0}^{\otimes n}$  under a Haar random unitary $W_i$. In this case one has
\begin{align}\label{eq:ortho-enhanced-vanish}
    h_{\thv}^{(2)}(\rho_i)=|\langle \Phi^+ |(W_i\otimes W_i ) |0\rangle ^{\otimes 2n}\,|^2\,,
\end{align}
which depends on the choice of $W_i$.  
Overall, we have shown that choosing $\widetilde{O}(\thv)=\dya{\Phi^+}$ leads to a model invariant under $\mathbb{O}(d)$ that satisfies Eq.~\eqref{eq:classification_reversal}, and hence, that can perform noisy classification (according to Definition~\ref{def:accuracy}) of the states in the time-reversal dataset.
\end{proof}

Theorem~\ref{theo:orthogonal-quadratic} shows that measuring the Bell state allows us to do classification. However, this does not solve the scaling issue discussed earlier. 
Indeed, as proven in Appendix~\ref{app:F}, the predictions values of the model given in Eq.~\eqref{eq:ortho-enhanced-vanish} still concentrate exponentially close to zero as a function of the number of qubits. 
This implies that we still need an exponential number of experiment repetitions to accurately classify the data. 
However, as we now show, this problem can be overcome if we slightly modify the task at hand, from the classification of time-reversal-symmetric states to the classification of time-reversal-symmetric \textit{dynamics}.

\begin{figure}[t]
\centering
\includegraphics[width=1\columnwidth]{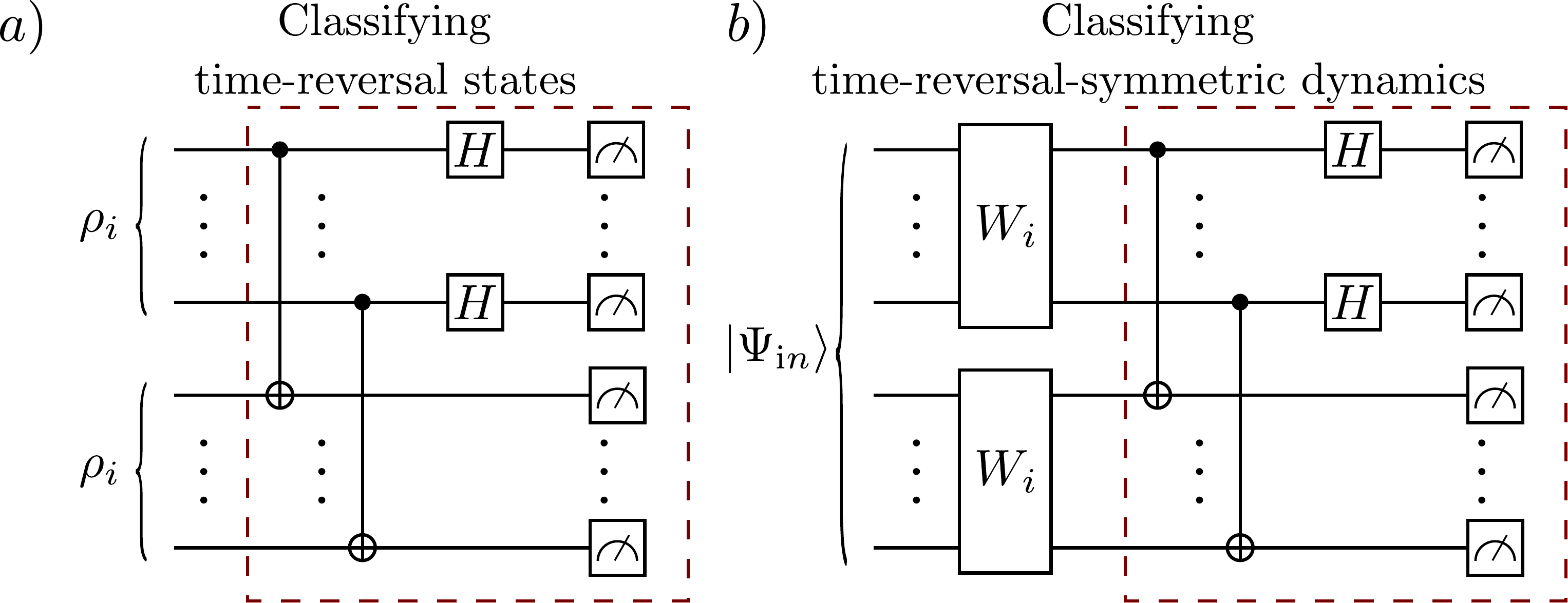}
\caption{\textbf{Circuits used in quantum-enhanced experiments.}  a) When classifying time-reversal states, the dataset is composed of states $\rho_i$ obtained by evolving a fiduciary real-valued state with an orthogonal unitary or with a Haar random unitary.  b) When classifying time-reversal-symmetric dynamics, the dataset is composed of orthogonal and Haar random unitaries $W_i$. Here we are free to choose the $2n$-qubit initial state $\ket{\Psi_{\rm in}}$ that will be evolved under the action of $W_i^{\otimes 2}$. In both panels we have indicated with a red dashed box the circuit for implementing a Bell-basis measurement. In particular, an all-zero measurement outcome corresponds to the probability of measuring $\ket{\Phi^+}$.  }
\label{fig:orthogonal-circuits}
\end{figure}

For this new task, rather than being given states, we assume instead access to the unitaries used to produce these states. The corresponding dataset has the form $\{W_i,y_i\}_{i=1}^N$, with
\begin{equation}
    y_i=
    \begin{cases}
        1 & \text{if $W_i\in \mathbb{O}(d)$,}\\
        0 & \text{if $W_i\in \mathbb{U}(d)$,
        }
    \end{cases}
\end{equation}
which has the same two symmetry groups $\mathfrak{G}_1=\mathbb{O}(d)$ and $\mathfrak{G}_0=\mathbb{U}(d)$ as before.

As shown in Fig.~\ref{fig:orthogonal-circuits}(b), the main advantage of this scenario is that we are now allowed to initialize the $2n$--qubit register to any global state $\ket{\Psi_{\tin}}$, and to simultaneously evolve the first and the second $n$ qubits according to the \textit{same} unitary $W_i$. To capture this additional freedom, we consider models in a new hypothesis class defined as:
\begin{hypothesis}~\label{def:hyp_class_Bell}
We define the Hypothesis Class $\HC_2$, computable in a quantum-enhanced experiment, as composed of functions of the form
\small
\begin{align}
    h_{\thv}(W)=\Tr[U(\thv)(W^{\otimes 2})\dya{\Psi_{\tin}} (W^{\otimes 2})\ad U\ad(\thv) O]\,,
\end{align}
\normalsize
where $U(\thv)$ is quantum neural network acting on the $2n$ qubits,  $O$ is a Hermitian operator, and $\ket{\Psi_{\tin}}$ is an initial state on $2n$ qubits.
\end{hypothesis}

In this context, we can still use Proposition~\ref{prop:no-post-processing-symmetries-two} to show that the following theorem holds.
\begin{theorem}\label{theo:orthogonal-quadratic-2}
Let $h_{\thv}\in\HC_2$ be a model in Hypothesis Class~\ref{def:hyp_class_Bell}, computable in a quantum-enhanced experiment.  There always exist quantum neural networks $U(\thv)$ and operators $O$, resulting in $\widetilde{O}(\thv)=\dya{\Phi^+}$ -- with $\ket{\Phi^+}$ being the Bell state on $2n$ qubits -- such that $h_{\thv}$ is $\mathbb{O}(d)$-invariant, but not $\mathbb{U}(d)$-invariant, that can perfectly classify the dynamics in the time-reversal dataset. The special choice of $\ket{\Psi_{\tin}}=\ket{\Phi^+}$ recovers the algorithm for classifying time-reversal-symmetric dynamics presented in~\cite{huang2021quantum}.
\end{theorem}

\begin{proof}
Recall from Proposition~\ref{prop:no-post-processing-symmetries-two} that $h_{\thv}$ is invariant under $\mathbb{O}(d)$, but not under $\mathbb{U}(d)$, if $\widetilde{O}(\thv)$ is in $\CC^{(2)}(\mathbb{O}(d))$ but not in $\CC^{(2)}(\mathbb{U}(d))$.  
Following the proof of Theorem~\ref{theo:orthogonal-quadratic}, we know that this can be achieved with the choice of $\widetilde{O}(\thv)=\dya{\Phi^+}$. Moreover, a straightforward calculation shows that if we choose $\ket{\Psi_{\tin}}=\ket{\Phi^+}$, we have
\begin{align}
    h_{\thv}(W_i)=1\,, \quad \forall\, W_i \text{ with label } y_i=1\,,
\end{align}
recovering the algorithm in~\cite{huang2021quantum}.
On the other hand, $h_{\thv}(W_i)$ is $W_i$-dependent and will concentrate around zero if $y_i=0$ (see Appendix). 
This means that the model outputs a value of $1$ if the unitary has label $y_i=1$, and outputs a value of $0$ (with high probability) if the unitary has label $y_i=0$. Thus, the models in Hypothesis Class~\ref{def:hyp_class_Bell} can perform perfect classification (according to Definition~\ref{def:accuracy}) of time-reversal-symmetric dynamics.
\end{proof}

As shown in the proof of Theorem~\ref{theo:orthogonal-quadratic-2}, now the model gives non-overlapping predictions for the data in different classes, meaning that we can now perform classification with $\OC(1)$ experiment repetitions. This is in contrast to the model defined in Theorem~\ref{theo:orthogonal-quadratic}, which requires an exponential number of experiments for accurate classification. This illustrates how QML models capable of achieving a quantum advantage naturally emerge in our framework as $\mathfrak{G}$-invariant models.

\subsection{Multipartite entanglement dataset}
\label{sec:entanglement}
In this section, we consider the more involved task of classifying pure quantum states according to the amount of multipartite entanglement they possess. Entanglement has been shown to be a fundamental resource~\cite{horodecki2009quantum,gigena2020one} in quantum information, quantum computation and quantum sensing~\cite{barrett2002nonsequential,gigena2017bipartite,ekert1991quantum,gisin2002quantum,ekert1998quantum,datta2005entanglement,chalopin2018quantum,beckey2020variational,cerezo2021sub}. Hence, its study and characterization is quintessential for quantum sciences. 

Here, we recall that  entanglement is relatively well understood for bipartite pure quantum states (e.g., via the Schmidt decomposition for pure states~\cite{nielsen2000quantum}), and that group-invariance arguments have been previously used to characterize entanglement in bipartite mixed states~\cite{terhal2000entanglement,vollbrecht2001entanglement}. However, the same cannot be said for the multipartite entanglement~\cite{walter2016multipartite,sawicki2014convexity,sawicki2012critical,macikazek2018asymptotic}. In this case, the entanglement complexity scales exponentially with the number of parties and there is no unique measure to quantify it. Thus, we employ our framework to not only obtain $\mathfrak{G}$-invariant QML models -- that can accurately classify multipartite entangled states -- but also to better understand the unique nature of multipartite entanglement. In this context, we also recall that the presence of publicly available datasets, such as the NTangled dataset~\cite{schatzki2021entangled}, composed of quantum states with varying amounts of multipartite entanglement, makes this an extremely rich application for our framework and for benchmarking QML models.

Let $\EC$ be a multipartite entanglement measure satisfying $\EC(\rho)\in[0,1]$, with $\EC(\rho)=0$ if the state is separable, and $\EC(\rho)>0$ if the state contains multipartite entanglement between the $n$ qubits (for instance, see the entanglement measures in Refs.~\cite{walter2013entanglement,prove2021extending,foulds2020controlled,beckey2021computable}). The multipartite  entanglement dataset is of the form  $\SC=\{(\rho_i,y_i)\}_{i=1}^N$, where
\begin{equation}
    y_i = 
    \begin{cases}
        1 & \text{if $\EC(\rho_i)=b>0$,}\\
        0 & \text{if $\EC(\rho_i)=0$.}
    \end{cases}
\end{equation}
Here, the symmetry group $\mathfrak{G}$ associated with the data in both classes is the Lie group $\mathfrak{G}=\bigotimes_{j=1}^n \mathbb{U}(2)$, with an associated Lie algebra $\mathfrak{g}=\bigoplus_{j=1}^n \mathfrak{su}(2)$. This is due to the fact that local unitaries $\bigotimes_{j=1}^n V_j$ do not change the multipartite entanglement in a quantum state. 

\subsubsection{Conventional experiments}

Since computing the entanglement typically requires evaluating a non-linear function of the quantum state~\cite{horodecki2009quantum}, it is expected that models in conventional experiments will not be able to classify the states in this dataset.  This intuition can be confirmed with the following theorem: 
\begin{theorem}\label{theo:entanglement-conventional}
Let $h_{\thv}^{(1)}\in\HC_1$ be a model in Hypothesis Class~\ref{def:hyp_class_1}, computable in a conventional experiment. There exists no quantum neural network $U(\thv)$ and operator $O$ such that $h_{\thv}^{(1)}$ is invariant under the action of $\bigotimes_{j=1}^n \mathbb{U}(2)$ and can classify (i.e., provide any relevant information about) the data in the multipartite entanglement dataset. 
\end{theorem}

\begin{proof}
First let us  verify that Propositions~\ref{prop:no-post-processing-symmetries} and~\ref{prop:no-post-processing-orthogonal} do not yield any adequate model for classification purposes. 
To identify the linear symmetries of $\mathfrak{G}$ required for the application of Proposition~\ref{prop:no-post-processing-symmetries},
we apply the Commutation Theorem for tensor products~\cite{rieffel1975commutation,mendl2009unital}, which  states that the commutant of a tensor product of operators is the tensor product of the commutants of each operator. 
Hence
\begin{align}
    \CC(\mathfrak{G})&={\rm span}\, (\{ \id_2^{\otimes n} \}) \,,
\end{align}
where $\id_2$ denotes the $2\times 2$ identity. This results in the choice  $\widetilde{O}(\thv)=\lambda \id$ ($\lambda \in \Rbb$) and constant model predictions ($\forall \rho_i$, $h^{(1)}_{\thv}(\rho_i)=\lambda$) that cannot distinguish between states.
Additionally, one can verify that  the orthogonal complement of $\mathfrak{g}$ is trivial:
\begin{align}
\mathfrak{g}^\perp &=\{\vec{0}_2^{\otimes n}\}\,,
\end{align}
with $\vec{0}_2$ the $2\times 2$ null matrix, such that models designed under Proposition~\ref{prop:no-post-processing-orthogonal} would also result in uninformative constant value predictions.

Hence, using Propositions~\ref{prop:no-post-processing-symmetries} and~\ref{prop:no-post-processing-orthogonal} to obtain  $\mathfrak{G}$-invariant models from Hypothesis Class~\ref{def:hyp_class_1} (with $k=1$)  will lead to trivial models that cannot classify the states in the multipartite entanglement dataset. Following a similar argument as the one developed in the last part of the proof of Theorem~\ref{theo:purity-conventional}, one can  also verify that no other $\mathfrak{G}$-invariant models exist with $k=1$. 
Indeed, if $h_{\thv}^{(1)}\left(V \rho V\ad\right)$ is invariant for any $V=\bigotimes_{j=1}^d V_j$, it also has to be invariant when uniformly averaged over every $V_j$ in $\mathbb{U}(2)$. Performing this averaging, we obtain
\begin{align}
     \mathbb{E}_{\mathfrak{G}}\left[h_{\thv}^{(1)}\left(V \rho V\ad\right)\right]&=\frac{\Tr[\rho]\Tr[O]}{d}\,,
\end{align}
for  $V=\bigotimes_{j=1}^d V_j$. The only way for  $h_{\thv}^{(1)}(\rho)$ to be equal to $\Tr[\rho]\Tr[O]/d$ for any state $\rho$ is to have $O\propto\id/d$ or $O=\vec{0}$, that is, solutions already  covered by Propositions~\ref{prop:no-post-processing-symmetries} and~\ref{prop:no-post-processing-orthogonal}. 
\end{proof}

\subsubsection{Quantum-enhanced experiments}
Let us first consider the case of $k=2$ copies in Eq.~\eqref{eq:linear_cost}. We can show that the following theorem holds.

\begin{theorem}\label{theo:entanglement}
Let $h_{\thv}^{(2)}\in\HC_1$ be a model in Hypothesis Class~\ref{def:hyp_class_1}, computable in a quantum-enhanced experiment.  There always exist quantum neural networks $U(\thv)$ and operators $O$, resulting in $\widetilde{O}(\thv)={\rm span}\big(\bigotimes_{j=1}^n\{\id_4^{(j)},SWAP^{(j)}\}\big)$, such that $h^{(2)}_{\thv}$ is invariant under the action of $\bigotimes_{j=1}^n \mathbb{U}(2)$ and can perfectly classify the data in the multipartite entanglement dataset. Here,  $\id_4^{(j)}$ denotes the $4\times 4$ identity matrix on the $j$-th qubit of  each copy of $\rho$, and $SWAP^{(j)}$ denotes the operator that swaps the $j$-th qubits of  each copy of $\rho$. There exist special choices of $\widetilde{O}(\thv)$ which recover all the multipartite entanglement measures proposed in Refs.~\cite{brennen2003observable,meyer2002global,rungta2001universal,bhaskara2017generalized,beckey2021computable,carvalho2004decoherence,wong2001potential}.
\end{theorem}

\begin{proof}
From Proposition~\ref{prop:no-post-processing-symmetries} we know that $h_{\thv}^{(2)}$ will be $\mathfrak{G}$-invariant  if $\widetilde{O}(\thv)$ is a quadratic symmetry of $\mathfrak{G}$. Here we can again invoke the Commutation Theorem for tensor products to obtain the space of these symmetries:
\begin{align}\label{eq:quadratic-entanglement}
    \CC^{(2)}(\mathfrak{G})={\rm span}\bigg(\,\bigotimes_{j=1}^n \left\{\id_4^{(j)},SWAP^{(j)} \right\}\bigg)\,,
\end{align}
where $\id_4^{(j)}$ denotes the $4\times 4$ identity matrix acting on the $j$-th qubit of each of the two copies of $\rho$, and where $SWAP^{(j)}$ denotes the operator that swaps the $j$-th qubits of the copies of $\rho$. Note that $\CC^{(2)}(\mathfrak{G})$ is spanned by $2^n$ elements, meaning that there exists an exponentially large  freedom in choosing $\widetilde{O}(\thv)$. Evidently, some choices of $\widetilde{O}(\thv)$ will not be useful for characterizing multipartite entanglement. For instance,  $\widetilde{O}(\thv)=\bigotimes_{j=1}^n \id_4^{(j)}$ leads to trivial model's predictions $h_{\thv}^{(2)}(\rho)=\Tr[\rho^{\otimes 2}]=1$ for any state $\rho$. Similarly, $\widetilde{O}(\thv)=\bigotimes_{j=1}^n SWAP^{(j)}=SWAP$ leads to $h_{\thv}^{(2)}(\rho)=\Tr[\rho^{\otimes 2}SWAP]=\Tr[\rho^2]=1$. 

On the other hand, there are choices for $\widetilde{O}(\thv)$ that can indeed characterize entanglement. For instance,
\begin{equation}\label{eq:impurity}
    \widetilde{O}(\thv)=2\left(\id - SWAP^{(j)}\otimes\id^{(\overline{j})} \right)\,,
\end{equation}
with $\id^{(\overline{j})}$ the identity on all qubits but the $j$-th ones leads to $h_{\thv}^{(2)}(\rho)=2(1-\Tr[\rho_j^2])$ which is the impurity of $\rho_j=\Tr_{\overline{j}}[\rho]$, i.e., the impurity of the reduced state on the $j$-qubit, and a bipartite entanglement measure across the cut $j$-th qubit / rest. Averaging over each of the $n$ qubits, i.e.,
\begin{equation}
    \widetilde{O}(\thv)=\frac{2}{n}\sum_{j=1}^n \left( \id -  SWAP^{(j)}\otimes\id^{(\overline{j})} \right)\,,
\end{equation}
recovers the multipartite entanglement measures of~\cite{brennen2003observable,meyer2002global}. 

Notably, the result in Eq.~\eqref{eq:impurity} can be further generalized as follows.  First, let us define $S=\{1,2,\ldots, n\}$  as the set of integers indexing each qubit, and let $P(S)$ be its \textit{power set} (i.e., the set of subsets of $S$, with cardinality $|P(S)| = 2^n$). Defining the operator
\begin{equation}
\widetilde{O}(\thv)=2\left(\id - \bigotimes_{j\in Q}SWAP^{(j)}\otimes\id^{(\overline{j})}\right)\,,
\end{equation}
for any $Q\in P(S)\backslash\{\emptyset\}$, leads to the generalized version of the Concurrence measure for multipartite pure states in~\cite{rungta2001universal,bhaskara2017generalized}. Even more generally, for any choice of $Q$, the operator
\begin{equation}\label{eq:con-cent}
    \widetilde{O}(\thv)=\id - \frac{1}{2^{|Q|}}\bigotimes_{j\in Q} \left(\id_4^{(j)} + SWAP^{(j)}\right)\otimes\id^{\overline{Q}}\,,
\end{equation}
where $\overline{Q}=S\backslash Q$, leads to the Concentratable Entanglement family of multipartite entanglement measures introduced in~\cite{beckey2021computable} (see Proposition 3 in~\cite{beckey2021computable}), where the special case $Q=S$ also leads to the measure of~\cite{carvalho2004decoherence}. On the other hand, 
\begin{equation}\label{eq:con-cent-2}
    \widetilde{O}(\thv)=\id - \frac{1}{2^{n}}\bigotimes_{j=1}^n \left(\id_4^{(j)} - SWAP^{(j)}\right)\,,
\end{equation}
leads to the $n$-tangle measure~\cite{wong2001potential} (see Proposition 5 in~\cite{beckey2021computable}). 

Since several of the previous choices for $\widetilde{O}(\thv)$ lead to entanglement monotones, the model's output will be different for data in different classes. Hence, one can perfectly classify the data in the multipartite entanglement dataset.
\end{proof}

The results in Theorem~\ref{theo:entanglement}  showcase how  Propositions~\ref{prop:no-post-processing-symmetries}--\ref{prop:no-post-processing-orthogonal-two} can lead to extremely powerful and non-trivial results. By simply imposing the $\mathfrak{G}$-invariance condition on the model one is able to naturally find an exponentially large manifold of solutions capable of classifying the states in the multipartite entanglement dataset. The latter leads to the intriguing possibility that $\CC^{(2)}(\mathfrak{G})$ contains solutions leading to new entanglement measures. 

Going further, one could also investigate the potential of models acting on more than $k=2$ copies, a prospect that has been largely unexplored. 
Here we know from Propositions~\ref{prop:no-post-processing-symmetries} that if $\widetilde{O}(\thv)$ is a $k$-th order symmetry, then the model $h_{\thv}^{(k)}$ is $\mathfrak{G}$-invariant under the action of $\bigotimes_{j=1}^n \mathbb{U}(2)$. Combining Eq.~\eqref{eq:k_symmetries_unitary} and the  Commutation Theorem leads to
\begin{align}\label{eq:quadratic-entanglement-3}
    \CC^{(k)}(\mathfrak{G})={\rm span}\, \bigg(\bigotimes_{j=1}^n S_k^{j}\bigg)\,,
\end{align}
where $S_k^{j}$ denotes the Symmetric group acting on the $k$ copies of the $j$-th qubit. Since the dimension of $\CC^{(k)}(\mathfrak{G})$ scales as $(k!)^n$, i.e., exponentially with $k$, the manifold of $\mathfrak{G}$-invariant models is likely to lead to a rich variety of entanglement measures.

\section{Discrete Group-Invariant Models}\label{section:discrete}
In the previous section we focused solely on situations where the symmetry group associated with the dataset was a unitary representation of some compact Lie group. 
Still, our formalism can be equally applied in the case of representations of discrete groups. 
Discrete groups are the relevant mathematical structure, for instance, when the quantum data is invariant under a finite set of permutations. 
This covers cases involving spatial invariance of condensed matter states on a lattice, or structural invariances in states of molecular systems. 
To illustrate such potential applications, we now address the task of classifying states belonging to a dataset with symmetry group $\mathfrak{G}=S_n$, i.e., the \textit{Symmetric group} consisting of all the $n!$ permutations over a set of $n$ indices.

\subsection{Graph isomorphism dataset}
In this section, we consider a dataset related to the so-called graph isomorphism problem, where the goal is to determine if two graphs are isomorphic. 
This classification task has a rich history in computational sciences~\cite{kobler2012graph}, and is known to be in the NP complexity class (although it has not been shown to be NP-complete). To solve this problem, several classical algorithms (with quasipolynomial complexity in the graph size~\cite{babai2016graph}), and also quantum heuristics~\cite{hen2012solving,gaitan2014graph,zick2015experimental,izquierdo2020discriminating} have been proposed. When using a quantum model for graph classification purposes, the first step is to encode graphs onto quantum states.  Here we take such encoding to be fixed and start by detailing how it is performed and how the ``graph isomorphism'' dataset is generated.

Recall that a graph is specified as $\mathcal{G} = (\mathcal{V} , \mathcal{E})$, where $\mathcal{V}$ is a collection of $n$ nodes, and $\mathcal{E}$ is a collection of edges. Two graphs $\GC$ and $\GC'$ are said to be isomorphic (and denoted as $\GC\cong \GC'$) if there exists a bijection between the sets of edges belonging to $\GC$ and $\GC'$. 
To build the dataset, we fix two reference non-isomorphic graphs  $\GC^0$ and $\GC^1$ ($\GC^0\not\cong\GC^1$), and generate graphs $\GC_i$ that are isomorphic to either $\GC^0$ or $\GC^1$.
To each of these graphs we assign labels $y_i=0$ or $y_i=1$, if it is isomorphic to $\GC^0$ or  $\GC^1$, respectively. 

Next, we encode these graphs $\{ \GC_i\}$ into quantum states $\{\rho_i\}$. This is achieved by evolving an initial $S_n$-invariant fiduciary state $\rho_{\rm in}$ (e.g., $\rho_{\rm in}=\dya{+}^{\otimes n}$)  as
\begin{equation}\label{eq:graph-states}
    \rho_i=W(\GC_i) \rho_{\rm in} W\ad(\GC_i)\,,
\end{equation}
where $ W(\GC_i)=e^{-i t H(\GC_i)}$. Here, $t>0$ is a fixed evolution time chosen such that the action of  $W(\GC^0)$ and $W(\GC^1)$ over $\rho_{\rm in}$ is different, and $H(\GC_i)$ is an Hamiltonian whose topology is that of the graph $\GC_i$. Specifically, we take it to be defined as
\begin{equation}\label{eq:graph-states-ham}
    H(\GC_i)=\sum_{(j,j')\in\mathcal{E}_i} Z_j Z_{j'}+\sum_{j\in\mathcal{V}_i}X_j\,,
\end{equation}
with $Z_j$ and $X_j$ denoting the Pauli $Z$ and $X$ operators acting on qubit $j$, respectively. 
We note that there exists more general ways to encode and process graph information in quantum states. In particular, the \textit{quantum graph convolutional neural network} introduced in Ref.~\cite{verdon2019quantumgraph} generalizes the unitary $W(\GC_i)$ used here and is specifically tailored to represent quantum systems that have a graph structure (see Appendix~\ref{app:QGCNN} for a detailed description).

\begin{figure}[t]
\centering
\includegraphics[width=.9\columnwidth]{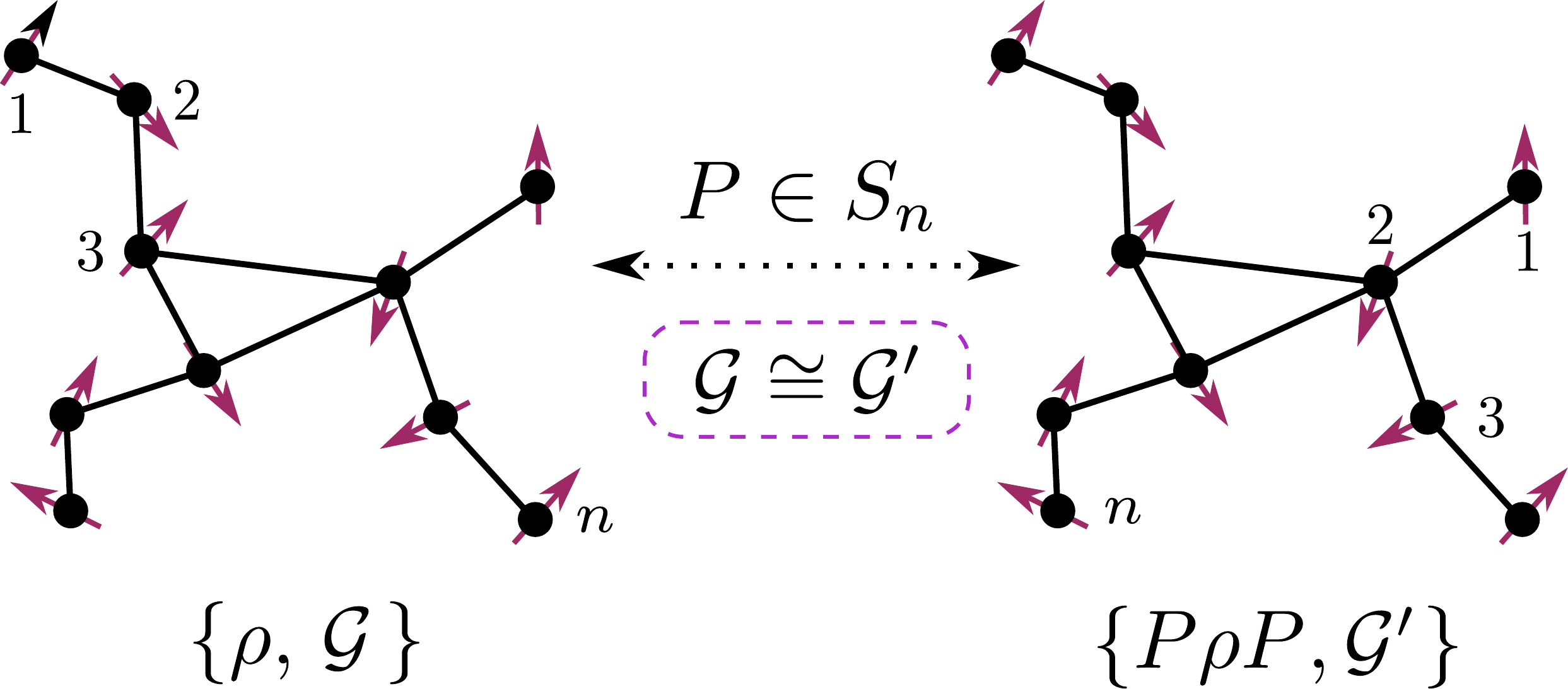}
\caption{\textbf{Permutation invariance in the graph isomorphism dataset.} Consider the state $\rho$ representing a quantum system of $n$-spins interacting under some Hamiltonian that follows the topology of an underlying graph $\GC$. By conjugating the state with an element $P\in S_n$, one obtains a new quantum state $P\rho P$ whose interaction graph $\GC'$ is isomorphic to $\GC$. That is, $\rho$ and $P\rho P$ have the same label in the dataset.}
\label{fig:graph}
\end{figure}

Taken together, the graph generation and the encoding in Equation~\eqref{eq:graph-states} allow us to define the graph isomorphism dataset as a collection $\SC=\{(\rho_i,y_i)\}_{i=1}^N$ of states $\rho_i$ with labels
\begin{equation}
    y_i = 
    \begin{cases}
        1 & \text{if $\GC_i\cong\GC^1$,}\\
        0 & \text{if $\GC_i\cong\GC^0$.}
    \end{cases}
\end{equation}
As shown in Fig.~\ref{fig:graph}, the states $\rho_i$ can be thought as representing an $n$-qubit quantum system whose interaction topology follows that of a graph $\GC_i$.  The symmetry group $\mathfrak{G}$ associated with both classes in the dataset is the Symmetric group $S_n$, as one can map states within the same class via the action of elements in $S_n$ (see Fig.~\ref{fig:graph}). Explicitly, let $P$ be an operator in $S_n$, and define the state $\rho_i'=P\rho_i P$, which can be expressed as $\rho_i'=P W(\GC_i)\rho_{\rm in} W\ad(\GC_i)\ad P$. 
Using the fact that $P e^{-i t H(\GC_i)}P=e^{-i t H(\GC_i')}$ (for a Hamiltonian  $H(\GC_i')=P H(\GC_i)P$ with $\GC_i'\cong\GC_i$) we have $\rho_i'=W(\GC_i')P \rho_{\rm in}P W\ad(\GC_i')$. 
Since $\rho_{\rm in}$ is $S_n$-invariant, then $P \rho_{\rm in}P=\rho_{\rm in}$, and given that $\GC_i'\cong \GC_i$, we conclude that the state $\rho_i'=W(\GC_i') \rho_{\rm in}W\ad(\GC_i')$ shares the same label as $\rho_i$.

Having defined the dataset of interest, we proceed to show that models in conventional experiments suffice to classify the data:
\begin{theorem}\label{theo:graphs}
Let $h_{\thv}^{(1)}\in\HC_1$ be a model in Hypothesis Class~\ref{def:hyp_class_1}, computable in a conventional experiment. There always exist quantum neural networks $U(\thv)$ and operators $O$, resulting in $\widetilde{O}(\thv)\in{\rm span}(A^{\otimes n})$ with $A$ in $ \mathbb{U}(2)$, such that $h_{\thv}^{(1)}$ is invariant under the action of $S_n$.
\end{theorem}

\begin{proof}
From Proposition~\ref{prop:no-post-processing-symmetries} we know that $h_{\thv}^{(1)}$ will be invariant under the action of $S_n$ if $\widetilde{O}(\thv)$ is a linear symmetry of $S_n$. Using the Schur-Weyl duality leads to~\cite{goodman2000representations}
\begin{equation}\label{eq:k_symmetries-ent}
    \CC(S_n)={\rm span}(\{A^{\otimes n}\,|\quad \forall A\in\Ubb(2)\}),
\end{equation}
meaning that the operator $\widetilde{O}(\thv)$ has to be a Hermitian linear combination of $n$-fold tensor products of single-qubit unitaries. Notably, the dimension of this solution manifold grows polynomially with $n$ as the dimension of $\CC(S_n)$  can be shown to follow the Tetrahedral numbers.
\end{proof}

Finally, it remains to identify an adequate $\widetilde{O}(\thv)$ among the space of operators that was just defined. 
This choice should be taken such that the model predictions are (maximally) distinct for states belonging to different classes, that is, $h_{\thv}^{(1)}(\rho_{i})\neq h_{\thv}^{(1)}(\rho_{i'})$ if $y_i\neq y_{i'}$. 
For an adequate parameterization of $\widetilde{O}(\thv)$, this search could be turned into an optimization task. 
Given that, by construction, the model predicts the same value for any states $\rho_i$ belonging to the same class, this optimization would only require a single representative state for each class.
Otherwise, one could employ heuristically-defined or physically-motivated operators. 
In particular, we highlight that the operators studied in~\cite{hen2012solving,gaitan2014graph,zick2015experimental,izquierdo2020discriminating} to distinguish non-isomorphic graphs belong to the family of $\widetilde{O}(\thv)$ yielded by Theorem~\ref{theo:graphs}. 

Overall, Theorem~\ref{theo:graphs} examplifies how our present framework can be readily applied to datasets with discrete symmetries. While studied here for the case of $\mathfrak{G}=S_n$, this can be specialized to subgroups of $S_n$, such as the group of reflexions, translations, or more subtle discrete symmetries which naturally arise in condensed-matter models and quantum chemistry problems.

\section{Equivariant Quantum Neural Networks}\label{sec:equivariance}

While the framework laid here provides the \textit{ultimate form} that a $\mathfrak{G}$-invariant model $h_{\thv}$ should adopt, it does neither prescribe how to actually parameterize the quantum neural networks $U(\thv)$, nor tell us how to choose the measurement operator $O$ that realizes  $\widetilde{O}(\thv)=U\ad(\thv) O U(\thv)$ such that it complies with our theorems. 

For this purpose, it is convenient to consider the action of the quantum neural network and the measurement process separately, and note that these can be generally though as concatenated maps. More generally, a constructive way to achieve $\mathfrak{G}$-invariance of general QML models starts by decomposing the model as
\begin{equation}\label{eq:decomp_model}
    h_{\thv} = \EC^{(M)}_{\thv^{(M)}} \circ  \cdots \circ \EC^{(1)}_{\thv^{(1)}} \,,
\end{equation}
i.e., as a composition of $M$ maps $\EC^{(m)}_{\thv^{(m)}}$, with integer labels $m=1,\ldots,M$, each parameterized by a subset of parameters $\thv^{(m)} \subset \thv$.
For instance, one such map could represent the action of a QNN (or of its individual constituents, i.e., its layers), the final measurement operation, steps of post-processing, or even the process of encoding classical data into quantum states in the first place.

Although imposing \textit{invariance} at the map level effectively results in global invariance of the model, this is quite restrictive. A more relaxed approach towards the construction of group invariant models involves the concept of \textit{equivariance}~\cite{cohen2016group,kondor2018generalization,bronstein2021geometric,castelazo2021quantum} which is now defined.
\begin{definition}[$\mathfrak{G}$-equivariance]\label{def:equivariance}
    Given a group $\mf{G}$ with an action on spaces $\AC$ and $\BC$, a function $\EC: \AC \mapsto \BC $ is called $\mathfrak{G}$-equivariant if it commutes with the action of the group
    \begin{equation}
        \EC(V \cdot x)=V \cdot \EC(x)\,,
    \end{equation}
    for all elements $V\in \mathfrak{G}$ and inputs $x\in\AC$.
\end{definition}
That is, a function is $\mf{G}$-equivariant if group-shifting the input ($x \mapsto V\cdot x$) produces a group-shifted output ($\EC(x) \mapsto V\cdot \EC(x)$). It can be verified that (i) $\mathfrak{G}$-equivariance of the intermediary maps ($m < M)$ along with (ii) $\mathfrak{G}$-invariance of the final map ($m=M$) is sufficient to ensure $\mathfrak{G}$-invariance of the composed model.
Intuitively, equivariance permits the propagation of the action of elements of $\mathfrak{G}$ up to the final map, so that the symmetries get preserved.

Accordingly, a model $h_{\thv}$ belonging to Hypothesis Class~\ref{def:hyp_class_1} could be split in terms of the transformation of the input state and the final measurement. 
That is, $h_{\thv}(\rho)=\EC^{(2)} \circ \EC^{(1)}_{\thv}(\rho)$ with $\EC^{(1)}_{\thv}(\rho)= U(\thv) \rho^{\otimes k} U^\dagger(\thv)$ and $\EC^{(2)}(\rho)= \Tr [O\rho^{\otimes k}]$. Thus, the model $h_{\thv}$ will be $\mathfrak{G}$-invariant if $\EC^{(2)}$ is invariant and $\EC^{(1)}_{\thv}$ is equivariant. As before, invariance of  $\EC^{(2)}$ can be achieved by choosing the measurement operator $O$ to be a Hermitian operator satisfying one of the Propositions~\ref{prop:no-post-processing-symmetries}--\ref{prop:no-post-processing-orthogonal-two}. On the other hand, equivariance of $\EC^{(1)}_{\thv}$ can be achieved in a way very close in essence to Proposition~\ref{prop:no-post-processing-symmetries}. Specializing Definition~\ref{def:equivariance} to the case of unitary maps acting on $k$ copies of $\rho$, we see that equivariance of such QNN (or of its layers) is equivalent to requiring that
\begin{equation}
    [U(\thv), V^{\otimes k}] =0,\quad\forall V \in \mathfrak{G}\,.
\end{equation}
That is, a unitary $U(\thv)$ is $\mf{G}$-equivariant if $U(\thv)\in\CC^{(k)}(\mathfrak{G})$ (it belongs to the $k$-th commutant of $\mf{G}$). For instance, for graph classification, the quantum graph convolutional neural network, presented in Appendix~\ref{app:QGCNN}, can be  verified to be be equivariant under the action of $S_n$.

Overall, approaching $\mf{G}$-invariance through the decomposition in Eq.~\eqref{eq:decomp_model} has the advantage of modularity -- equivariant maps could be more easily identified and reused across different models -- and also allows for the study of more general models than the ones in Hypothesis Class~\ref{def:hyp_class_1}. 
For instance additional steps of post-processing, or of encoding of classical data into quantum states, can be described as additional maps to be composed, and readily fit in such framework. Finally, although in this manuscript we have exclusively focused on classification tasks, i.e., where the model outputs scalars, one can consider the more general setting of a model producing operator valued outputs (i.e., in the case of quantum generative modeling \cite{romero2017quantum,verdon2019quantum,zhu2019training}). In such a case, one would be interested in global equivariance of the model, rather than invariance.

\section{Conclusions}\label{section:conclusions}

In this work, we presented a theoretical framework to design QML models that, by construction, respect the symmetries of a group $\mathfrak{G}$ associated to the dataset. This approach has several benefits~\cite{bronstein2021geometric}: it is more data efficient, it reduces the model's search space (less parameters), it often leads to better generalization, and the classification accuracy is robust under perturbations drawn from the symmetry group.  Our main contributions are as follows. First, in Propositions~\ref{prop:no-post-processing-symmetries}--\ref{prop:no-post-processing-orthogonal-two} we  leveraged properties from representation theory to determine the conditions that lead to $\mathfrak{G}$-invariance. These results constitute guidelines for designing group-invariant models and were used to show how models such as those in Hypothesis Class~\ref{def:hyp_class_1}  can be made $\mathfrak{G}$-invariant and accurately solve certain supervised learning tasks. We then showcased the power of our framework for several QML tasks, where we find how embedding symmetry information into the model allows us to recover in an elegant and formal way several algorithms from the literature that were heuristically obtained, or that were obtained through trial-and-error.

As a first application, we addressed the task of classifying pure states from mixed states in \textit{conventional} and \textit{quantum-enhanced} experiments (i.e., experiments with and without access to a quantum memory). For this case, the symmetry group is the unitary group, since applying a unitary to a quantum state preserves its spectral properties.  Theorem~\ref{theo:purity-conventional} showed that there exist no conventional experiments that are  $\mathfrak{G}$-invariant and that can classify the data in the purity dataset. However, by allowing the QML model to coherently act on two copies of each state in the dataset we showed  in Theorem~\ref{theo:unitary-quadratic} that there exist models that are $\mathfrak{G}$-invariant and can classify the data. These are based on taking the expectation value of the SWAP operator, which naturally appeared through the Schur-Weyl duality as an element of the Symmetric group.

The second task we considered was that of classifying time-reversal symmetric states from Haar random states.   In Theorems~\ref{th:time-reversal-conventional} and~\ref{theo:orthogonal-quadratic} we showed that models in both conventional and quantum-enhanced experiments can be used to distinguish such states. Surprisingly, we recovered the well-known Bell basis measurement scheme for detecting time-reversal, where the Bell measurement operator appeared naturally as an element of the basis of the Brauer algebra.  Moreover, by the seemingly innocuous change of allowing the model to access the unitaries that prepare the states in the dataset (rather than the states themselves), we can obtain in Theorem~\ref{theo:orthogonal-quadratic-2} the model used in  Ref.~\cite{huang2021quantum} to show that quantum-enhanced experiments can classify the data with exponentially less experiments than conventional experiments. This connects our work with recent research showing that QML models are capable of exponential advantages in some tasks of data classification. 

We then applied our framework to a dataset composed of pure states with different amounts of multipartite entanglement. Here, the symmetry group preserving entanglement is the $n$-fold direct product of the local unitary group. This example proved the power of our framework as we showed that all the entanglement measures of  Refs.~\cite{brennen2003observable,meyer2002global,rungta2001universal,bhaskara2017generalized,beckey2021computable,carvalho2004decoherence,wong2001potential} are special cases of the family of $\mathfrak{G}$-invariant models defined in Theorem~\ref{theo:entanglement}. Moreover, we conjectured that allowing the QML model to access more than two copies of each quantum state and measuring the expectation value of local permutation operators can lead to new entanglement measures. Interestingly, we recently became aware of the work in Ref.~\cite{liu2022detecting} where it was shown such  expectation values can indeed detect the presence entanglement.  

Additionally, we showed how our results extend beyond continuous Lie groups, by studying a problem of classification in a quantum graph isomorphism dataset. That is, we addressed the task of determining if a given graph-encoded quantum state belongs to one isomorphism class or the other. In this case, the symmetry group associated with the dataset is the Symmetric group. In Theorem~\ref{theo:graphs}, we identified $\mathfrak{G}$-invariant models capable of classifying the data in such graph-isomorphism dataset.

Our results take one of the first  steps towards a general theory of QML models with sharp geometric priors based on the dataset symmetries. Since our work is  inspired by the theory and success of geometric deep learning, we envision that soon enough the field of \textit{geometric quantum  machine learning} will be a thriving and exciting field.

\section{Outlook}\label{section:outlook}

Here we overview some questions  left unanswered by our results, and propose different paths forward.

\subsection{Equivariance}
As detailed in Section~\ref{sec:equivariance}, the concept of equivariance in quantum neural networks may play a central role when building models that respect the symmetries of a given dataset. While a few  examples of equivariant quantum neural networks have been proposed, such as the Quantum Convolutional Neural Network~\cite{cong2019quantum} which respects translational symmetry, the Quantum Convolutional Graph Neural Network~\cite{verdon2019quantumgraph} which respects $S_n$-symmetry in graphs, the $\mathbb{U}(d)$-equivariant ansatz of~\cite{zheng2021speeding}, or the graph automorphism group-invariant ansatz in~\cite{sauvage2022building}, it is worth noting that these are the exception to the rule. Most quantum neural networks in the literature are not equivariant, and do not use information about symmetries in their design. Hence, much work remains to be done in the path towards general equivariant architectures, especially to guarantee that they have circuit depth and connectivity requirements compatible with near-term quantum hardware.

\subsection{Trainability: Expressibility and gradient magnitudes}
Arguably, one of the main threats to the trainability of QNNs are Barren Plateaus (BPs), a phenomenon by which gradients along the parameter landscape become exponentially concentrated around zero as the system size grows~\cite{holmes2021connecting,mcclean2018barren,cerezo2020cost,sharma2020trainability,thanasilp2021subtleties,arrasmith2021equivalence}.  In the presence of BPs, an exponential number of measurement shots is required to correctly identify a minimizing direction on the landscape. Given such limitations, understanding the conditions that lead to their presence has been the subject of extensive work~\cite{grant2019initialization,cerezo2020cost,pesah2020absence,volkoff2021large,larocca2021diagnosing}.

Naively, one would be tempted to choose QNNs to be highly expressive so that good approximations of the relevant unitary transformation can be achieved. Nevertheless, Ref.~\citep{holmes2021connecting} unveiled a connection between the expressibility of an ansatz and the magnitudes of the gradients: highly expressive ansätze were shown to exhibit BPs, suggesting that expressibility should be limited to give room for trainability. 
Later on, Ref.~\citep{larocca2021diagnosing} pointed towards the Lie closure of the gate generators of an ansatz as a measure of its \textit{ultimate} expressibility. 
Most importantly, when the dimension of such Lie closure grows exponentially with the system size (as is the case for problem-agnostic architectures such as the harware efficient ansatz~\cite{kandala2017hardware,mcclean2018barren,cerezo2020cost}) there exists some critical number of layers beyond which barren plateaus are known to dominate the parameter landscapes. Finally, we note that the size of the Lie closure has also been related to the number of parameters needed to overparametrize a QNN~\citep{larocca2021theory}. 
Therefore, given a fixed number of parameters we expect, in general, less expressive ansatz to have more favourable landscapes.
Overall, all these results point towards the importance of reducing as much as possible (in a sensible way) the expressivity of QNNs.

In this context, building models with strong geometric priors such as equivariant QNNs constitute a sensible choice. By constraining the expressibility of the ansatz to the relevant region only, these symmetry-based proposals emerge as goldilocks candidates for trainability-aware ansatz design. While the exact improvement in tranability will be certainly problem dependent, there is already evidence that equivariant QNNs do indeed lead to better performance and trainability in several archetypal near-term algorithms~\cite{sauvage2022building}.

\subsection{Generalization}
Complementary to its trainability, the ability of a model to generalize to unseen data is key to its applicability in realistic scenarios. 
While errors evaluated on a training dataset are the main metric when training a model, its practical success should be gauged when applied to new testing data.
The generalization error quantifies the gap between training and testing errors.
In the realm of QML, recent results~\cite{caro2021generalization} have shown that such generalization error is upper bounded by a quantity scaling as $\sqrt{T/N}$ where $T$ denotes the number of parameters, and $N$ the number of training data.
As such, given a fixed training dataset, reducing the expressibility of a model (by means of equivariant QNNs with appropriate geometric priors), and thus the number of free parameters, is expected to yield better model generalization.

\subsection{Quantum advantage}

The gold standard for quantum machine learning models, and for quantum algorithms in general, is being able to solve a given task faster that any classical method. As exemplified by our main results (see Sec.~\ref{sec:time-reversal}), the concept of $\mathfrak{G}$-invariance is not tied to that of computational advantage, as there exists $\mathfrak{G}$-invariant models capable, but also incapable, of achieving  a quantum advantage. Hence, it will be fundamental to determine the key features that lead to models with favourable scalings.

\subsection{More general models and learning scenarios}

In this work we considered QML models of the form in Hypothesis Class~\ref{def:hyp_class_1}. However, these are not the most general models one can have. For instance, the ability to perform non-trivial post-processing on the measurement outcomes~\cite{cincio2018learning} or employing randomized measurement techniques~\cite{elben2022randomized} can greatly increase the model's power and performance~\cite{huang2021provably}. 
We expect that the principles exposed here can be applied to more general settings opening up the possibility of obtaining $\mathfrak{G}$-invariance with methods beyond those  described in Propositions~\ref{prop:no-post-processing-symmetries}--\ref{prop:no-post-processing-orthogonal-two}. Moreover, while the concepts of $k$-th order symmetries and orthogonal complements played a key role in our derivation of $\mathfrak{G}$-invariant models, we expect that other properties will be needed to understand group-invariance in more general settings. 

Finally we highlight that we have mainly focused on supervised tasks of binary classification. Nevertheless, the ideas of $\mathfrak{G}$-invariance should also be applied to more general supervised learning scenarios (including regression problems), or to unsupervised learning scenarios.

\section{Acknowledgments}

We thank Robert Zeier for helpful and insightful discussions.  ML and PJC were supported by the U.S. Department of Energy (DOE), Office of Science, Office of Advanced Scientific Computing Research, under the Accelerated Research in Quantum Computing (ARQC) program. ML was also supported by the Center for Nonlinear Studies at LANL.   PJC and MC were also initially supported by the LANL ASC Beyond Moore's Law project.  FS was supported by the Laboratory Directed Research and Development (LDRD) program of Los Alamos National Laboratory (LANL) under project number 20220745ER. MC was supported by the LDRD program of LANL under project number 20210116DR.  This work was also supported by the Quantum Science Center (QSC), a National Quantum Information Science Research Center of the U.S. Department of Energy (DOE). X, formerly known as Google[x], is
part of the Alphabet family of companies, which includes Google, Verily, Waymo, and others (\url{www.x.company}).

\bibliography{quantum.bib}

\cleardoublepage

\onecolumngrid

\appendix

\begin{center}
	{\Large \bf Appendices} 
\end{center}

\setcounter{section}{0}
\setcounter{proposition}{0}
\setcounter{figure}{0}
\setcounter{corollary}{0}
\setcounter{definition}{0}
\setcounter{lemma}{0}

\bigskip

Here we present the additional details and proofs for some of the  results in the main text.

\section{ Hermitian part of the commutant}\label{app:A}

Let us first show that the following lemma holds

\begin{lemma}
Let $\CC^{( k)}(\mathfrak{G})$ denote the $k$-th order symmetries of   $\mathfrak{G}\subseteq\mathbb{U}(d)$. Then, for any matrix $A$ in $\CC^{( k)}(\mathfrak{G})$, its Hermitian conjugate $A\ad$ is also in $\CC^{( k)}(\mathfrak{G})$.
\end{lemma}

\begin{proof}~\label{lemma:Hermiticity}
Let $A$ be a matrix in $\CC^{( k)}(\mathfrak{G})$. From Definition~\ref{def:symmetries} we know that $A$ is such that
\begin{equation}
[A,V^{\otimes k}]=0\,,\quad \forall V\in \mathfrak{G}\,.
\end{equation}
The previous equation can be explicitly written as $AV^{\otimes k}=V^{\otimes k} A $, or equivalently, as
\begin{equation}
A=V^{\otimes k} A (V\ad)^{\otimes k}\,,\quad \forall V\in \mathfrak{G}\,.
\end{equation}
Taking the conjugate transpose on each side leads to 
\begin{equation}
A\ad=V^{\otimes k} A\ad (V\ad)^{\otimes k}\,,\quad \forall V\in \mathfrak{G}\,,
\end{equation}
which, by Definition~\ref{def:symmetries}, implies that  $A\ad$ is also in $\CC^{( k)}(\mathfrak{G})$.
\end{proof}

As a consequence of Lemma~\ref{lemma:Hermiticity}, one can always associate a Hermitian operator to any matrix in $\CC^{(k)}(\mathfrak{G})$. Namely, if $A\in\CC^{(k)}(\mathfrak{G})$ is Hermitian, then nothing needs to be done. But if $A$ is not Hermitian, one can create the operators $A+A\ad$ and $i(A-A\ad)$ which are Hermitian and which also belong to $\CC^{(k)}(\mathfrak{G})$.

\setcounter{proposition}{1}
\section{Proof of Proposition~\ref{prop:no-post-processing-orthogonal}}\label{app:B}

Here we prove Proposition~\ref{prop:no-post-processing-orthogonal} of the main text, which we recall for convenience.

\begin{proposition}\label{prop:no-post-processing-orthogonal-SM}
Let $h_{\thv}^{(k)}\in\HC_1$ be a model in  Hypothesis Class~\ref{def:hyp_class_1}. Then, let $\mathfrak{G}$ be the symmetry Lie group associated with the dataset, and let $\mathfrak{g}\subseteq\mathfrak{u}(d)$ be its Lie algebra with $i\id\in\mathfrak{g}$. 
The model will be $\mathfrak{G}$-invariant when $\rho \in i \mathfrak{g}$ and $\widetilde{O}(\thv)\in{\rm span}(\{A_j\otimes A_{\overline{j}}\}_j )$. Here, $A_j\in i\mathfrak{g}^{\perp}$ is a Hermitian operator acting on the $j$-th copy of $\rho$ and $A_{\overline{j}}$ is an operator acting on all copies of $\rho$ but the $j$-th one.
\end{proposition}

\begin{proof}

Let us first consider the case when $i\id\in\mathfrak{g}$, and when $\widetilde{O}(\thv)=A_j\otimes A_{\overline{j}}$ for a given $j$.  If $\rho \in i \mathfrak{g}$ then
\begin{equation}
\rho=\sum_\mu r_\mu h_\mu  \quad \text{with} \quad h_\mu\in i\mathfrak{g} \quad \text{and} \quad r_\mu\in\mathbb{R}\,.
\end{equation}
Similarly, if $A_j\in i\mathfrak{g}^{\perp}$ we can write
\begin{equation}
    A_j=\sum_{\mu'} a_{\mu'} g_{\mu'}\,, \quad \text{with} \quad  g_{\mu'}\in i\mathfrak{g}^\perp \quad \text{and} \quad a_\mu\in\mathbb{R}\,.
\end{equation}
It then follows from Definition~\ref{def:orthogonal-complement} that
\begin{equation}\label{eq:proof-prop-2-0}
    \Tr[\rho A_j]=\sum_{\mu,\mu'}r_\mu a_{\mu'}\Tr[h_\mu g_{\mu '}]=0\,,
\end{equation}
since each term $\Tr[h_\mu g_{\mu '}]=0$.

From Eq.~\eqref{eq:proof-prop-2-0} we know that the expectation value of $\widetilde{O}(\thv)$ over $\rho^{\otimes k}$ is
\begin{align}\label{eq:proof-prop-2-1}
    h_{\thv}^{(k)}(\rho)=\Tr[\rho^{\otimes k}\widetilde{O}(\thv) ] = \Tr[\rho A_j]\Tr_{\overline{j}}[\rho^{\otimes(k-1)}A_{\overline{j}}]=0\,, 
\end{align}
where $\Tr_{\overline{j}}$ is the  trace over all qubits except those in the $j$-th copy of $\rho$. Equation~\eqref{eq:proof-prop-2-0} shows that $h_{\thv}^{(k)}=0$  for all states $\rho$ having support exclusively on $i \mathfrak{g}$. Since a Lie algebra is closed under the action of its Lie group, then it follows that $V\rho V\ad\in i\mathfrak{g}$ for all $V\in\mathfrak{G}$.  Thus, with a similar argument, we have 
\begin{equation}
h_{\thv}^{(k)}(V\rho V\ad)=  \Tr[V \rho V\ad A_j]\Tr_{\overline{j}}[(V\rho V\ad)^{\otimes(k-1)}A_{\overline{j}}]=0
\end{equation}
for all $V\in\mathfrak{G}$, where we used again Definition~\ref{def:orthogonal-complement}. The latter shows that $h_{\thv}^{(k)}$ is $\mathfrak{G}$ invariant.

Finally, we can generalize the previous results to the case when  $\widetilde{O}(\thv)\in{\rm span}(\{A_j\otimes A_{\overline{j}}\}_j) $. Now we can expand 
\begin{equation}
    \widetilde{O}(\thv)=\sum_j o_j A_j\otimes A_{\overline{j}}\,.
\end{equation}
If $A_j\in i\mathfrak{g}^{\perp}$ $\forall j$, we finally find 
\begin{align}
   h_{\thv}^{(k)}(V\rho V\ad) =\Tr[(V\rho V\ad)^{\otimes k}\widetilde{O}(\thv) ]=\sum_j o_j \Tr[V \rho V\ad A_j]\Tr_{\overline{j}}[(V \rho V\ad)^{\otimes(k-1)}A_{\overline{j}}]=0\,,
\end{align}
which completes the proof.
\end{proof}

We here finally note that if $i\id \in \mathfrak{g}^{\perp}$, then we can obtain $\mathfrak{G}$-invariance when $\rho \in i \mathfrak{g}^{\perp}$ and $A_j\in i\mathfrak{g}$. The proofs follows similarly to that previously presented. Now we have 
\begin{align}
   h_{\thv}^{(k)}(V\rho V\ad) =\Tr[(V\rho V\ad)^{\otimes k}\widetilde{O}(\thv) ]=\sum_j o_j \Tr[ \rho V\ad A_j V]\Tr_{\overline{j}}[(V \rho V\ad)^{\otimes(k-1)}A_{\overline{j}}]=0\,.
\end{align}
Here we need to use the fact that $\Tr[ \rho V\ad A_j V]\in i\mathfrak{g}$ for all $V\in\mathfrak{G}$ (a Lie algebra is closed under the action of its associated Lie group). Moreover, we have used Definition~\ref{def:orthogonal-complement} to show that  $\Tr[ \rho V\ad A_j V]=0$.

\section{Proof of Propositions~\ref{prop:no-post-processing-symmetries-two} and~\ref{prop:no-post-processing-orthogonal-two}}\label{app:C}

Here we prove Propositions~\ref{prop:no-post-processing-symmetries-two} and~\ref{prop:no-post-processing-orthogonal-two}. Since their proofs follow from those of Propositions~\ref{prop:no-post-processing-symmetries} and~\ref{prop:no-post-processing-orthogonal}, respectively, we simply mention the main differences.

\begin{proof}

First, we recall that in Propositions~\ref{prop:no-post-processing-symmetries-two} and~\ref{prop:no-post-processing-orthogonal-two} we consider the case when the data with different labels have different symmetry groups, denoted as  $\mathfrak{G}_0$ and $\mathfrak{G}_1$.  Moreover, each symmetry group will have their own $k$-th order symmetries, and for the case of Lie groups, their associated Lie algebra and orthogonal complements. 

From Proposition~\ref{prop:no-post-processing-symmetries} we know that a model will be $\mathfrak{G}$-invariant if $\widetilde{O}(\thv)$ belongs to $\CC^{( k)}(\mathfrak{G})$. Since here there are two symmetry groups, we know that the models will be $\mathfrak{G}_0$-invariant if $\widetilde{O}(\thv)$ belongs to $\CC^{( k)}(\mathfrak{G}_0)$. Evidently, the model will also be $\mathfrak{G}_1$ invariant if  $\widetilde{O}(\thv)$ is also in $\CC^{( k)}(\mathfrak{G}_1)$. Thus, to guarantee $\mathfrak{G}_0$-invariance, but not necessarily $\mathfrak{G}_1$-invariance, one needs $\widetilde{O}(\thv)$ to  belong to $\mathfrak{G}_0$, but not to $\mathfrak{G}_1$. This constitutes the proof of Proposition~\ref{prop:no-post-processing-symmetries-two}.

The proof of Proposition~\ref{prop:no-post-processing-orthogonal-two} follows similarly from that of Proposition~\ref{prop:no-post-processing-symmetries-two}. Here, the model will be $\mathfrak{G}_0$- and $\mathfrak{G}_1$-invariant, if  $\widetilde{O}(\thv)\in{\rm span}(\{A_j\otimes A_{\overline{j}}\}_j) $,  $\rho \in i \mathfrak{g}_0,i \mathfrak{g}_1$,  and $A_j\in i\mathfrak{g}_0^{\perp},i\mathfrak{g}_1^{\perp}$. In addition, the model will be $\mathfrak{G}_{0(1)}$-invariant but not necessarily  $\mathfrak{G}_{1(0)}$-invariant when $\rho \in i \mathfrak{g}_0,i \mathfrak{g}_i$  and $A_j\in i\mathfrak{g}_0^{\perp}$ but $A_j\not\in i\mathfrak{g}_1^{\perp}$. 
\end{proof}

\section{Ancilla-based models for the purity dataset}\label{app:D}

In the main text we considered the task of classifying  the data in the purity dataset with models in Hypothesis Class~\ref{def:hyp_class_1}, i.e., with models of the form $h_{\thv}(\rho)=\Tr[U(\thv)\rho^{\otimes k} U\ad(\thv)O]$. As we saw in Theorem~\ref{theo:purity-conventional} of the main text, there are no such models with $k=1$ that allow for classification. On the other hand,  Theorem~\ref{theo:unitary-quadratic} shows that models with $k=2$ can indeed classify the data according to their purity. In this case, $\widetilde{O}(\thv)$ has to be the SWAP operator (up to some additive and multiplicative constants). 

The previous raises the issue of how to efficiently evaluate the expectation value of the SWAP operator. Taking inspiration from the Hadamard Test, which computes the expectation value of a unitary by controlling its action with an ancillary qubit, we envision a new family of ancilla-based hypothesis. Here, one appends to the QNN an extra qubit that is used along the two copies of $\rho$, so that  the $2n+1$ qubit state $(\dya{0}\otimes\rho\otimes\rho)$ is fed into a quantum neural network that acts globally on all qubits, and only measures the ancilla qubit. This  defines the following Hypothesis class. 
\begin{hypothesis}~\label{def:hyp_class_ancilla}
We define the Hypothesis Class $\HC_3$, computable in a quantum-enhanced experiment, as composed of functions of the form
\begin{align}
    h_{\thv}(\rho)=\Tr[U(\thv)(\dya{0}\otimes\rho\otimes\rho) U\ad(\thv)O_A]\,,
\end{align}
where $U(\thv)$ is a quantum neural network acting on $2n+1$-qubits, and $O_A=(O\otimes\id\otimes\id)$ with $O$ a one-qubit Hermitian operator acting on the ancilla qubit.
\end{hypothesis}

The models in Hypothesis Class $\HC_3$ now should be invariant under the action of $\id\otimes V\otimes V$ for any $V$ in $\mathbb{U}(d)$.
In the spirit of Proposition~\ref{prop:no-post-processing-symmetries} and defining $\widetilde{O}_A(\thv)=U\ad(\thv)O_AU(\thv)$, we know that this can be readily achieved when $[\widetilde{O}_A(\thv),\id\otimes V\otimes V]=0$ for all $V\in\mathbb{U}(d)$, that is, when $\widetilde{O}_A\in i\mf{u}(2)\otimes\CC^{(2)}(\mathfrak{G})$. This results in the following Theorem.
\begin{theorem}\label{theo:unitary-quadratic-ancilla}
Let $h_{\thv}\in \HC_3$ be a model in Hypothesis Class~\ref{def:hyp_class_ancilla}, computable in a quantum-enhanced experiment. There always exist quantum neural networks $U(\thv)$ and operators $O_A$, resulting in $\widetilde{O}_A(\thv)=A\otimes S$, where  $A\ket{0}=\ket{0}$ and $S\in{\rm span}(\{\id \otimes \id,SWAP\})$ with non-zero component in $SWAP$, such that $h_{\thv}$ is invariant under the action of $\mathbb{U}(d)$  and can perfectly classify the data in the purity dataset. The special choice of $\widetilde{O}_A(\thv)=Z\otimes SWAP$ corresponds to the operator measured in the Swap Test~\cite{buhrman2001quantum} and in the ancilla based algorithm of~\cite{cincio2018learning}. 
\end{theorem}

\begin{proof}
We first recall that the models in $\HC_3$ will be $\mathfrak{G}$-invariant if $h_{\thv}(V\rho V \ad)=h_{\thv}(\rho)$ for all $V\in\mathbb{U}(d)$, i.e., when 
\begin{equation}\label{eq:to-satisfy-2}
    \Tr[(\id\otimes V\otimes V)(\dya{0}\otimes\rho\otimes\rho) (\id\otimes V\ad \otimes V\ad) \widetilde{O}_A(\thv)]=\Tr[(\dya{0}\otimes\rho\otimes\rho) \widetilde{O}_A(\thv)]\,.
\end{equation}
This holds when
\begin{equation}\label{eq:to-satisfy-3}
    [\widetilde{O}_A(\thv),(\id\otimes V\otimes V)]=0\,.
\end{equation}
Eq.~\eqref{eq:to-satisfy-3} is satisfied for the choice of $\widetilde{O}_A(\thv)=A\otimes S$ with $A$ an operator acting on the ancillary qubit and $S$ an operator in $\CC^{(2)}(\mathbb{U}(d))={\rm span}(S_2)$ with $S_2=\{\id\otimes \id,SWAP\}$ a representation of the Symmetric group of two elements. Then, replacing $\widetilde{O}_A(\thv)$ by $A\otimes S$  in the left-hand-side of~\eqref{eq:to-satisfy-2} leads to 
\begin{align}
    \Tr[(\id\otimes V\otimes V)(\dya{0}\otimes\rho\otimes\rho) (\id\otimes V\ad \otimes V\ad) \widetilde{O}_A(\thv)]&=\Tr[(\id\otimes V\otimes V)(\dya{0}\otimes\rho\otimes\rho) (\id\otimes V\ad \otimes V\ad) A\otimes S)]\nonumber\\
    &=\Tr[(\dya{0}\otimes\rho\otimes\rho)  (A\otimes S)]\nonumber\\
    &=\Tr[\dya{0} A]\Tr[\rho\otimes \rho S]\,.\nonumber
\end{align}
In the second inequality we have used the fact that, $S$ commutes with $V\otimes V$, while in the third line we have simply separated the trace over the different subsystems. Since replacing $\widetilde{O}_A(\thv)$ by $A\otimes S$  in the left-hand-side of~\eqref{eq:to-satisfy-2} leads to $\Tr[\dya{0}]\Tr[\rho\otimes \rho S]$, we can see that the model will be $\mathfrak{G}$-invariant iff $A\ket{0}=\ket{0}$, i.e., if $A$ has $\ket{0}$ as an eigenvector with eigenvalue equal to one.

Finally, we remark that a direct calculation with the circuits in Fig.~\ref{fig:ancilla}  verifies that both circuits satisfy  $U_1\ad(Z\otimes\id\otimes\id)U_1=U_2\ad(Z\otimes\id\otimes\id)U_2=Z\otimes SWAP$.  Hence, we recover the operator measured in the Swap Test~\cite{buhrman2001quantum} and in the ancilla based algorithm of~\cite{cincio2018learning}.
\end{proof}

Let us analyze here the remarkable fact that the special choice $A=Z$ leads to the \textit{exact} operator measured in two distinct ancilla-based circuits computing the purity: in the Swap Test~\cite{buhrman2001quantum} and in the algorithm discovered in~\cite{cincio2018learning}. For completeness, these two circuits are shown in Fig.~\ref{fig:ancilla} for the case when $\rho$ is a single qubit state. While both circuits end-up computing the purity of $\rho$  they do so by implementing distinct  unitaries. This can be seen by evaluating the Schmidt rank across the ancilla/data qubits cut of the two circuits.  This Schmidt rank is found to be $2$ for the unitary $U_1$ displayed in Fig.~\ref{fig:ancilla}(a), and $3$ for the unitary $U_2$ displayed in Fig.~\ref{fig:ancilla}(b). 
This indicates that both circuits are fundamentally different in the sense that there is no local operations that map one to the other~\cite{cincio2018learning}.
Still, one can verify that
\begin{equation}
    U_1\ad(Z\otimes\id\otimes\id)U_1=U_2\ad(Z\otimes\id\otimes\id)U_2=A\otimes SWAP\,,
\end{equation}
where $A=Z$. Hence, our results shed new light to the connection between these two circuits.

\begin{figure}[t]
\centering
\includegraphics[width=1\linewidth]{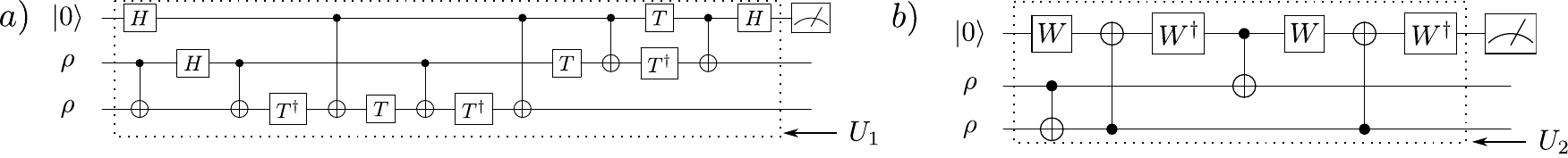}
\caption{\textbf{Ancilla-based circuits for computing the purity.} a) The circuit $U_1$ corresponds to the canonical SWAP test which was compiled into gates native to the IBM’s superconducting qubit devices~\cite{smolin1996five}. 
Here $H$ and $T$ denotes the Hadamard and $\pi/8$ phase gates respectively. 
At the end of the circuit one measures the expectation value of the Pauli $Z$ operator on the ancilla qubit. b) The circuit for $U_2$  corresponds to the ancilla-based algorithm for computing purity discovered in~\cite{cincio2018learning} trough a machine learning subroutine  designed to minimize the number of CNOTs required. (see also~\cite{bilkis2021semi}). Here $W=T\ad H$. At the end of the circuit one measures the expectation value of the Pauli $Z$ operator on the ancilla qubit. Both circuits satisfy the conditions in Theorem~\ref{theo:unitary-quadratic} as $U_1\ad(Z\otimes\id\otimes\id)U_1=U_2\ad(Z\otimes\id\otimes\id)U_2=Z\otimes SWAP$.   }
\label{fig:ancilla}
\end{figure}

\section{Classifying with Eq.~\eqref{eq:classification_reversal}}\label{app:E}

In the main text we argued that when the classification task has two symmetry groups, associated with the two different classes, then one can classify the data if there exists a $\mathfrak{G}_1$-invariant model $h_{\thv}^{(k)}\in\HC_1$ in the Hypothesis Class~\ref{def:hyp_class_1} such that
\begin{equation}\label{eq:classification_reversal_SM}
\begin{split}
    &h_{\thv}^{(k)}(\rho_i)=
    c\quad \text{if $y_i=1$}\,,\\
    &h_{\thv}^{(k)}(\rho_i)\in[b_1,b_2] \quad \text{if $y_i=0$}\,.
\end{split}
\end{equation}
When $c\not\in[b_1,b_2]$ one can readily use this model for unambiguous classification: when the model returns a value of $c$ (different from $c$) one would assign a label of $y=1$ ($y=0$) which corresponds to the true label of the state to be classified. 
However, when $c\in[b_1,b_2]$ we can still perform classification, but at the cost of misclassifying some of the states. Indeed, there will now be cases where one measures a value of $h_{\thv}^{(k)}(\rho_i)=c$ despite the fact that the underlying state has true label $y_i=0$, but assign it a label $y=1$. The probability of such misclassification event is quantified in the following.
\begin{lemma}\label{lemma:probab}
Let $P(0|c)$ be the probability of misclassification which happens when the true label of a state $\rho_i$ is $y_i=0$ given a model value $h_{\thv}^{(k)}(\rho_i)=c$. Accordingly, let $P(c|0)$ be the probability that the model takes a value of $c$ when the data has label $y_i=0$. Assuming equal probability of sampling states belonging to each of the two classes, we have
\begin{equation}\label{eq:lemma-probab-eq}
    P(0|c) = \frac{P(c |0)}{1 + P(c |0)}\,.
\end{equation}
\end{lemma}
Lemma~\ref{lemma:probab} shows that the probability of misclassification will remain small as long as the probability that $h_{\thv}^{(k)}(\rho_i)=c$ given $y_i=0$ is small. Let us now provide a proof for Lemma~\ref{lemma:probab}.

\begin{proof}
We know from Bayes theorem that 
\begin{equation}\label{eq:probab-1}
    P(0|c) = \frac{P(0)P(c|0)}{P(c)}\,.
\end{equation}
According to the law of total probability we have $P(c)=P(c|1)P(1)+P(c|0)P(0)$, where $P(0)$ and $P(1)$ denote the probability of sampling a state $\rho_i$ with label $y_i=0$ and $y_i=1$, respectively. 
Assuming the same representation of each label in the dataset, i.e., $P(0)=P(1)=1/2$, we can rewrite Eq.~\eqref{eq:probab-1} as  
\begin{equation}\label{eq:probab-2}
    P(0|c) = \frac{\frac{1}{2}P(c |0)}{\frac{1}{2}\big(P(c|1)+P(c|0)\big)} =\frac{P(c|0)}{P(c|1)+P(c|0)}\,.
\end{equation}
Finally, recalling that when $y_i=1$ the model outcome is always $c$, we have $P(c|1)=1$ such that we recover~\eqref{eq:lemma-probab-eq}.
\end{proof}

Going further, we can bound this probability of misclassification if we know the expectation value and variance of $h_{\thv}^{(k)}(\rho_i)$ for states with label $y_i=0$. First, note that $P(c |0)\geq 0$, such that, according to~\eqref{eq:lemma-probab-eq}, $P(0|c) \leq P(c |0)$. Next, let us define $X$ to be the random variable corresponding to the QML model output, i.e., $X=h_{\thv}^{(k)}(\rho_i)$ for a random state $\rho_i$. We denote as $\langle X\rangle_{0}=E_{y_i=0}[h_{\thv}^{(k)}(\rho_i)]$ the expectation value of this variable conditioned on the sampled states to have label $y_i=0$. 
Assuming that this expectation value is greater than $c$ (without loss of generality)  and defining $\delta =\langle X \rangle_{0} - c >0$, we find via Cantelli's inequality that
\begin{equation}
    P(X-\langle X\rangle\geq \delta)\leq \frac{\Var_0[X]}{\Var_0[X]+\delta^2}\,,
\end{equation}
where $\Var_0[X]=\langle X^2\rangle_0-\langle X\rangle_0^2$ is the variance of $X$ when sampling states with label $y_i=0$. It follows that 
\begin{equation}\label{eq:probab-4}
P(0|c) \leq P(c |0)\leq P(X-\langle X\rangle\geq \delta) \leq \frac{\Var_0[X]}{\Var_0[X]+\delta^2}\,,
\end{equation}
showing that for large separation $\delta$ relative to the variance  $\Var_0[X]$, the misclassification error will be small. 
Finally note that, when taking into account additive errors $\epsilon$ when estimating the output of the model, one would assign a label $1$ for model values estimated in a range $C=[c-\epsilon, c+\epsilon]$. In this case misclassification would arise when the model value is estimated in $C$ despite the true label of the state being $0$, and Equations~\eqref{eq:lemma-probab-eq} and~\eqref{eq:probab-4} could readily be extended to this scenario.

\section{Concentration results for time-reversal datasets}\label{app:F}

In the previous section we have seen that one can classify the states in the dataset via Eq.~\eqref{eq:classification_reversal_SM}, and that for well separated expectation values for the two classes, the misclassification probability can be small. 
In other cases, however, it can happen that $\langle X\rangle_{0}=E_{y_i=0}[h_{\thv}^{(k)}(\rho_i)]\sim c$ and that one would need a large number of experiment repetitions to guarantee an accurate classification. Here we see that for some of the models considered in the main text, such issue arises.

\subsection{Conventional experiments}

We recall from Theorem~\ref{th:time-reversal-conventional} that if $O$ is  a purely complex operator, then 
\begin{equation}
    h_{\thv}^{(1)}(\rho_i)=0\,, \quad \text{$\forall \rho_i$ such that $y_i=1$}\,,
\end{equation}
We now show that values of $h_{\thv}^{(1)}(\rho_i)$, for $y_i=0$, exponentially concentrates around $c=0$, such that an exponential number of shots would be needed for classification~\cite{chen2021exponential,aharonov2022quantum,huang2021quantum}.

First, let us recall that the Haar random states $\rho_i$ with $y_i=0$ are obtained by evolving a fiduciary real-valued initial state $\rho_{\tin}$ according to a Haar random unitary $W_i$:
\begin{equation}
    \rho_i=W_i \rho_{\tin} W_i\ad\,.
\end{equation}
Thus, we can evaluate the expectation value $\langle X\rangle_{0}=E_{y_i=0}[h_{\thv}^{(k)}(\rho_i)]$ by averaging the  model predictions over the unitary group. That is, $\langle X\rangle_{0}=E_{Haar}[h_{\thv}^{(1)}(W_i \rho_i W_i\ad)]$. This can be analytically derived via symbolical integration with respect to the Haar measure on a unitary group~\cite{puchala2017symbolic}. For any $W\in \mathbb{U}(d)$ the following expressions are valid for the first two moments of the distribution:  
\small
\begin{equation}\label{eq:delta}
\begin{aligned}
    \int_{\mathbb{W}(d)} d\mu(U)w_{\vec{i}_1\vec{j}_1}w_{\vec{i}_2\vec{j}_2}^*&=\frac{\delta_{\vec{i}_1\vec{i}_2}\delta_{\vec{j}_1\vec{j}_2}}{d}\,,   \\
\int_{\mathbb{W}(d)} d\mu(U)w_{\vec{i}_1\vec{j}_1}w_{\vec{i}_2\vec{j}_2}w_{\vec{i}_1'\vec{j}_1'}^{*}w_{\vec{i}_2'\vec{j}_2'}^{*}&=\frac{\delta_{\vec{i}_1\vec{i}_1'}\delta_{\vec{i}_2\vec{i}_2'}\delta_{\vec{j}_1\vec{j}_1'}\delta_{\vec{j}_2\vec{j}_2'}+\delta_{\vec{i}_1\vec{i}_2'}\delta_{\vec{i}_2\vec{i}_1'}\delta_{\vec{j}_1\vec{j}_2'}\delta_{\vec{j}_2\vec{j}_1'}}{d^2-1}
-\frac{\delta_{\vec{i}_1\vec{i}_1'}\delta_{\vec{i}_2\vec{i}_2'}\delta_{\vec{j}_1\vec{j}_2'}\delta_{\vec{j}_2\vec{j}_1'}+\delta_{\vec{i}_1\vec{i}_2'}\delta_{\vec{i}_2\vec{i}_1'}\delta_{\vec{j}_1\vec{j}_1'}\delta_{\vec{j}_2\vec{j}_2'}}{d(d^2-1)}\,,
\end{aligned}
\end{equation}
\normalsize
where $w_{\vec{i}\vec{j}}$ are the matrix elements of $W$. Assuming $d=2^n$, we use the notation $\vec{i} = (i_1, \dots i_n)$ to denote a bitstring of length $n$ such that $i_1,i_2,\dotsc,i_{n}\in\{0,1\}$. 

A straightforward calculation leads to 
\begin{equation}
    \langle X\rangle_{0}=\frac{\Tr[O]}{d}=0\,.
\end{equation}
Here we have used the fact that $O$ is Hermitian and purely imaginary, and hence has no support on the identity operator. Moreover, we can also find that 
\begin{equation}
    \Var_0[X]=\langle X^2\rangle_{0}=\frac{\Tr[O^2]\Tr[\rho_{\tin}^2]}{d^2-1}-\frac{\Tr[O^2]}{d(d^2-1)}= \frac{d \Tr[\rho_{\tin}^2]}{d^2-1}-\frac{1}{d(d^2-1)}\,,
\end{equation}
where in the last inequality we have further assumed that $\Tr[O^2]=\id$ (as this is the case when $O$ is a tensor product of an odd number of Pauli $Y$ operators). Hence, we can see that $\Var_0[X]\in\OC(\frac{1}{2^n})$\, and thus that the values of the QML model (when sampling Haar random states) concentrate exponentially around its mean of zero.

\subsection{Quantum-enhanced experiments}

We start by  recalling that when classifying the states in the time-reversal dataset, we found that
\begin{align}
    h_{\thv}^{(2)}(\rho_i)
    &=|\langle \Phi^+ | 0\rangle ^{\otimes 2n}|^2=\frac{1}{d^2}\,,
\end{align}
for all $\rho_i$ such that $y_i=1$, and that
\begin{align}\label{eq:ortho-enhanced-vanish-SM}
    h_{\thv}^{(2)}(\rho_i)=|\langle \Phi^+ |(U_i\otimes U_i ) |0\rangle ^{\otimes 2n}|^2\,,
\end{align}
for all $\rho_i$ such that $y_i=0$, where $U_i$ is the Haar random unitary that generates the state $\rho_i$. Here we have assumed that  $\rho_{\tin }=\dya{0}^{\otimes n}$. An explicit calculation of the expectation value leads to 
\begin{equation}
    \langle X\rangle_{0}=\frac{2}{d(d+1)}\,,
\end{equation}
which shows that for large number of qubits  $\langle X\rangle_{0}\sim c\in\OC(\frac{1}{2^n})$.
Moreover, explicit calculation of the variance shows, as before, that  $\Var_0[X]\in\OC(\frac{1}{2^{2n}})$. Hence the  values of the QML model (for Haar random states) concentrate in a range exponentially close to zero.

\section{Quantum Graph Convolutional Neural Networks}~\label{app:QGCNN}

In this section we present additional details for the \textit{Quantum Graph Neural Networks} (QGNN) and \textit{Quantum Graph Convolutional Neural Networks} (QGCNN) architectures introduced in Ref.~\cite{verdon2019quantumgraph}. We start by introducing the general form of a QGNN, and then show how it can be specialized to particular permutation-invariant datasets. 
\subsubsection{General QGNN ansatz}
 The most general QGNN ansatz is defined as:
\begin{equation}\label{eq:QGNN}
    \hat{U}_{\textsc{qgnn}}(\bm{\eta}, \bm{\theta}) = \prod_{p=1}^P \left[ \prod_{q=1}^Q e^{-i\eta_{pq}\hat{H}_q(\bm{\theta})}\right],
\end{equation}
which consists of a sequence of evolutions under $Q$ different Hamilonians, repeated $P$ times, with $\bm{\eta}$ and $\bm{\theta}$ the variational parameters. The Hamiltonians $\hat{H}_q(\bm{\theta})$ can be taken to be any parameterized Hamiltonians with interaction topology corresponding to the graph $\mathcal{G} = \{\mathcal{V},\mathcal{E}\}$:
\begin{equation}\label{eq:Ham} \hat{H}_q(\bm{\theta})\! \equiv\!\!\!\!\!\! \sum_{\{j,k\}\in \mathcal{E}} \sum_{r\in\mathcal{I}_{jk}} W_{qrjk} \hat{O}^{(qr)}_{j}\otimes \hat{P}^{(qr)}_{k}+ \sum_{v\in \mathcal{V}} \sum_{r\in\mathcal{J}_{v}} B_{qrv} \hat{R}^{(qvr)}_{v}.\end{equation}
Note that here the $W_{qrjk}$ and $B_{qrv}$ are tensors of real-valued parameters, which are collected as one vector $\bm{\theta} \equiv \cup_{q,j,k,r}\{W_{qrjk}\}\cup_{q,v,r}\{B_{qrv}\}$. The operators $\hat{R}^{(qvr)}_{j}, \hat{O}^{(qr)}_{j}, \hat{P}^{(qr)}_{j}$ are Hermitian and act on the Hilbert space of the $j^\mathrm{th}$ node of the graph, while the sets $\mathcal{I}_{jk}$ and $\mathcal{J}_{v}$ are index sets for the terms corresponding to the edges and nodes of the graph $\mathcal{G}$, respectively. 

\subsubsection{Quantum Graph Convolutional Network (QGCNN)}

In order to ensure permutation invariance of the overall ansatz, we can enforce the generators of the QGNN to all be permutation invariant.
This is achieved by  tying some $\bm{\theta}$ parameters across vertex indices of the graph, we can then drop such indices, and the coefficients become: $W_{qrjk} \mapsto W_{qr}$ and $B_{qrv}\mapsto B_{qr}$, resulting in
\begin{equation}\label{eq:Ham_conv} \hat{H}_q(\bm{\theta})\! \equiv\!\!\!\!\!\! \sum_{\{j,k\}\in \mathcal{E}} \sum_{r\in\mathcal{I}} W_{qr} \hat{O}_j^{(qr)}\otimes \hat{O}^{(qr)}_{k}+  \sum_{v\in \mathcal{V}}\sum_{r\in\mathcal{J}} B_{qr} \hat{R}^{(qr)}_{v},\end{equation}
where $\mathcal{I}$ and $\mathcal{J}$ are index sets that are edge and node-independent. 
These Hamiltonian generators of our QGNN layers are invariant under any permutation of the node indices which preserve the edge structure, as boths sums' indices can be relabelled.
We call QGNN's of the form \eqref{eq:QGNN} with generators of the form \eqref{eq:Ham_conv} \textit{Quantum Graph Convolutional Neural Networks} (QGCNN's) and note that they are general forms of the so-called Quantum Alternating Operator Ansatze~\citep{hadfield2019quantum}, and of the Quantum Approximate Optimization Algorithm ~\citep{farhi2014quantum}.

\end{document}